\documentclass[11pt,a4paper]{amsart}
\usepackage{amssymb,dsfont}
\usepackage[left=25mm,top=30mm,bottom=25mm,right=25mm]{geometry}
\usepackage{cmap}            
\usepackage{hyperxmp}
\usepackage[utf8]{inputenc}
\usepackage[english]{babel}
\usepackage[fixlanguage]{babelbib}
\usepackage[pdfdisplaydoctitle = true,
      colorlinks = true,
      urlcolor = blue,
      citecolor = blue,
      linkcolor = blue,
      pdfstartview = FitH,
      pdfpagemode = UseNone,
      bookmarksnumbered = true,
      unicode=true]{hyperref} 
\usepackage{bm}
\usepackage{calrsfs}
\usepackage[all]{xy}
\usepackage{float}
\usepackage{tensor}

\newtheorem{thm}{Theorem}[section]

\newtheorem{lem}[thm]{Lemma}

\theoremstyle{definition}

\newtheorem{exam}[thm]{Example}

\theoremstyle{remark}
\newtheorem{rem}[thm]{Remark}

\numberwithin{equation}{section}

\newcommand{\RR}{\mathbb{R}}                

\newcommand{\espace}{\mathcal{E}}           
\newcommand{\body}{\mathcal{B}}             
\newcommand{\Cinf}{\mathrm{C}^{\infty}}     
\newcommand{\Met}{\mathrm{Met}}             
\newcommand{\Emb}{{\mathrm{Emb}^{\infty}}}             

\newcommand{\bepsilon}{{\bm{\epsilon}}}
\newcommand{\bgamma}{{\bm{\gamma}}}

\newcommand{\bpi}{\bm{\pi}}
\newcommand{\bsigma}{{\bm{\sigma}}}
\newcommand{\btau}{{\bm{\tau}}}
\newcommand{\btheta}{{\bm{\theta}}}
\newcommand{\bbeta}{\bm{\eta}}

\newcommand{\bxi}{{\pmb{\xi}}}

\newcommand{\norm}[1]{\left\Vert#1\right\Vert}

\newcommand{\set}[1]{\left\{#1\right\}}

\newcommand{\bq}{\mathbf{q}}                
\newcommand{\vol}{\mathrm{vol}}             
\newcommand{\bt}{\mathbf{t}}

\newcommand{\XX}{\mathbf{X}}

\newcommand{\bC}{\mathbf{C}}
\newcommand{\bS}{\mathbf{S}}
\newcommand{\bT}{\mathbf{T}}
\newcommand{\bP}{\mathbf{P}}

\newcommand{\msigma}{{\pmb{\mathfrak{S}}}}

\newcommand{\mpi}{{\pmb{\pi}}}
\newcommand{\bN}{\bm{N}}

\newcommand{\bb}{\mathbf{b}}
\newcommand{\bd}{\mathbf{d}}
\newcommand{\nn}{\bm{n}}
\newcommand{\vv}{\bm{v}}
\newcommand{\ww}{\bm{w}}
\newcommand{\xx}{\mathbf{x}}
\newcommand{\yy}{\mathbf{y}}

\newcommand{\VV}{\bm{V}}                
\newcommand{\UU}{\bm{U}}                
\newcommand{\uu}{\bm{u}}                
\newcommand{\pp}{p}                     
\newcommand{\ee}{\bm{e}}                
\newcommand{\ff}{\bm{f}}                
\newcommand{\bF}{\mathbf{F}}            

\newcommand{\id}{\mathrm{id}}
\newcommand{\Id}{\mathbf{Id}}

\DeclareMathOperator{\Lie}{L} %
\DeclareMathOperator{\tr}{tr} %
\DeclareMathOperator{\grad}{grad} %
\DeclareMathOperator{\dive}{div} %
\DeclareMathOperator{\2dots}{:}

\renewcommand{\vec}{\pmb}

\hypersetup{
    pdftitle={An intrinsic geometric formulation of hyper-elasticity, pressure potential and non-holonomic constraints},
    pdfauthor={Boris Kolev and Rodrigue Desmorat},
    pdfsubject={},
    pdfkeywords={},
    pdflang=fr
}
\begin{document}

\title[Pressure potential and non-holonomic constraints]{An intrinsic geometric formulation of hyper-elasticity, pressure potential and non-holonomic constraints}

\author{B. Kolev}
\address[Boris Kolev]{Université Paris-Saclay, ENS Paris-Saclay, CNRS,  LMT - Laboratoire de Mécanique et Technologie, 91190, Gif-sur-Yvette, France}
\email{boris.kolev@math.cnrs.fr}

\author{R. Desmorat}
\address[Rodrigue Desmorat]{Université Paris-Saclay, ENS Paris-Saclay, CNRS,  LMT - Laboratoire de Mécanique et Technologie, 91190, Gif-sur-Yvette, France}
\email{rodrigue.desmorat@ens-paris-saclay.fr}

\date{August 18, 2021}%
\begin{abstract}
  Isotropic hyper-elasticity, altogether with the equilibrium equations and the usual boundary conditions, are formulated directly on the body $\body$, a three-dimensional compact and orientable manifold with boundary equipped with a mass measure. Pearson–Sewell–Beatty pressure potential on the boundary is recovered, using the Poincaré formula. The existence of such a potential requires conditions, which are formulated as non-holonomic constraints on the configuration space.
\end{abstract}

\maketitle


\section{Introduction}
\label{sec:intro}

The finite strain theory for solids requires the comparison of the deformed configuration with a reference configuration, usually the initial unloaded structure. These configurations are described using mappings from some abstract manifold with boundary, the \emph{body} $\body$, into the \emph{ambient (usually Euclidean) space} $\espace$~\cite{TN1965,SH1997}.
As mentioned by Truesdell and Noll~\cite{Nol1959,Nol1972} and later by several authors~\cite{WT1973,Rou1980,Hau2002,Dim2011,Ber2012,GLM2014,Ste2015}, the body does not have, from the pure differential geometry point of view, to be embedded in space and identified with some reference configuration. This point was emphasized by Noll~\cite{Nol1972,Nol1978}, who called a formulation of continuum mechanics on the body $\body$ as \emph{intrinsic}\footnote{The term \emph{intrinsic} can have additional meanings in geometry:  for instance ``coordinate free'' or, in surface theory, ``depending only on the metric and not on the embedding''.} and by Rougée~\cite{Rou1980,Rou1991} who refers to it as an \emph{intrinsic Lagrangian framework}, whereas a formulation on a reference configuration is denominated as a \emph{standard Lagrangian approach}. Intrinsic formulations nowadays use modern tools in differential geometry~\cite{KOS2017,JdLE2018,EJdL2019,JdLE2020,SE2020}. In this direction, we insist on the fundamental role played by the \emph{manifold of all the Riemannian metrics on the body} --- introduced in solids mechanics by Rougée~\cite{Rou1980,Rou1997,Rou2006} and Fiala~\cite{Fia2004,Fia2011,Fia2016} --- in the formulation of hyper-elasticity (see \autoref{sec:hyper-elasticity}).

Even if not necessary \textit{a priori} in practice, one must be able to formulate the finite strain theory, including the boundary conditions and the principle of virtual power, entirely on the body $\body$, not necessarily embedded in space. Meanwhile, following Epstein and Segev \cite{ES1980,Seg1986}, we are led to consider virtual powers $\mathcal{W}$ as one-forms on the configuration space (the space of embeddings), and therefore to formulate continuum mechanics in the framework of differential geometry in infinite dimension~\cite{Arn1965,Ham1982,Olv1993,KM1997}.

\emph{Hyper-elasticity} with prescribed displacement and conservative loads is usually formulated on a reference configuration as a variational problem, with an explicit Lagrangian (a potential energy) depending on the deformation $\varphi$ --- of the reference configuration into the deformed one --- and on its first derivatives (see~\cite{Bal1976/77,Cia1988,MH1994}). The further question of the existence and the determination of a Lagrangian (a potential) for the prescribed pressure boundary condition in finite strain theory has arisen. A partial answer has been given by Pearson~\cite{Pea1956} when studying elastic beams instabilities, and later improved by Sewell~\cite{Sew1965,Sew1967} and Beatty~\cite{Bea1970} (see also~\cite{PC1991}). A simplified, less general, formulation has been summarized by Ball~\cite{Bal1976/77}. Contrary to the case of small strains and displacements, in finite strain theory, some extra conditions are required for a pressure potential to exist. Such conditions have been formulated on a reference configuration in~\cite{Bea1970}, with no clue of how to recast them on $\body$.

Formulating the problem in an intrinsic manner, meaning on the body, was initiated by Rougée~\cite{Rou1997,Rou2006}, who defined an elasticity law as a vector field on the manifold of Riemannian metrics on the body. However, the formulation of the boundary conditions on $\body$ rather than on a reference configuration in ambient space seems to remain an open problem in the general case~\cite{Seg1986,GLM2014,Fia2016}.

In the present work, we address this question for the following common boundary conditions: the standard \emph{dead load}, and the \emph{prescribed pressure}. The novelty is that we recast these boundary conditions on the body $\body$, rather than, as is usually done, on the reference configuration. To achieve this task, we were lead to recast the boundary conditions using the general formalism of differential forms on manifolds, bypassing expressions which are limited to calculus in affine $3$-space. Besides, to produce a potential for the virtual work of these surface forces, the crucial concept we have used, is an infinite dimensional extension of a linear operator $K$, that we have called the \emph{Poincaré integrator}. This operator was initially defined to prove the \emph{Poincaré lemma}, which states that a \emph{closed} form $\alpha$ on $\RR^{n}$, \textit{i.e.}, $d\alpha = 0$, is always \emph{exact}, \textit{i.e.}, that there exists $\beta$ such as $\alpha=d\beta$. The operator $K$ was built to satisfy the following formula (see~\autoref{sec:differential-forms})
\begin{equation*}
  dK+Kd=\id,
\end{equation*}
where $d$ is the exterior derivative. Indeed, writing
\begin{equation*}
  d(K\alpha) = \alpha - Kd\alpha,
\end{equation*}
then, we get immediately that $\beta=K\alpha$ is the sought solution of $\alpha=d\beta$ if $d\alpha=0$. This construction is the right framework to solve the problem of finding a potential $\beta = \mathcal{L}$ (a functional on the manifolds of embeddings) corresponding to some virtual work $\alpha = \mathcal{W}$ (a one-form on the manifolds of embeddings). The existence of such a potential requires however an integrability condition $d\alpha = 0$, which is not satisfied in the prescribed pressure case. Anyway, if $d\alpha \ne 0$ (a condition on a two-form), then, the condition $Kd\alpha = 0$ (a condition on a one-form) can be used to formulate geometric integrability conditions, that we call \emph{non-holonomic constraints}, since they involve both the embedding and its first variation. If this condition is satisfied, then $\mathcal{L} := K\mathcal{W}$ can be considered as a potential for the virtual work $\mathcal{W}$, for such admissible virtual displacements, \textit{i.e.}, satisfying $Kd\mathcal{W} = 0$. When applied to the virtual work of prescribed pressure terms, these non-holonomic conditions will appear to be slightly more general than the original Beatty conditions~\cite{Bea1970}. Furthermore, this proposed geometric formulation seems to be new and to have not yet been used so far for the problem of boundary conditions considered in this paper. It is moreover systematic and could be used for other mechanical problems.

The outline of the paper is as follows. In~\autoref{sec:kinematics}, we recall the general geometric framework of continuum mechanics. In~\autoref{sec:hyper-elasticity}, we present a synthetic and revisited version of Rougée geometric theory of elasticity. In~\autoref{sec:virtual-power}, we recast the virtual powers as one-forms on the space of embeddings (defined on the body). The general techniques issued from the Poincaré lemma and the definition of the Poincaré integrator are provided in~\autoref{sec:Lagrangian-formulations}. Here, we explain how to formulate Lagrangians (potentials) for surface forces (boundary conditions) and illustrate the proposed general method through the simple example of \emph{dead loads}. The potential obtained is close to the one derived in~\cite{GLM2014}, but with the notable difference that our formulation of the displacement uses the affine structure of space. The more complicated problem of \emph{prescribed pressure} is discussed in~\autoref{sec:prescribed-pressure}, both on the body and on the reference configuration. In addition, five appendices have been added to provide the required technical details.

\section{Kinematics of finite strains}
\label{sec:kinematics}

In continuum mechanics, the ambient space $\espace$ is represented by an oriented three-dimensional Euclidean affine space. Denoting by $\bq$ the Euclidean metric on $\espace$, it is better to consider this space as a Riemannian manifold $(\espace,\bq)$ and forget, as far as possible, this affine structure of space except the fact that the tangent bundle $T\espace=\espace \times E$ is trivial, where $E$ is the vector space associated to $\espace$ (called the \emph{translation space} in~\cite{Nol1978}).

The material medium is parameterized by a three-dimensional compact and orientable manifold with boundary, noted $\body$, the \emph{body}. This manifold $\body$ is equipped with a \emph{volume form} $\mu$, the \emph{mass measure}~\cite{TN1965} and is thus orientable.

\begin{rem}
  It is common, when possible, to refer to \emph{material coordinates}, \textit{i.e.}, defined on the body $\body$, as uppercase letters $\XX$ while \emph{spatial coordinates}, \textit{i.e.}, defined on the space $\espace$, are represented by lowercase letters $\xx$.
\end{rem}

A \emph{configuration} of a material medium is represented by a smooth orientation-preserving \emph{embedding} (particles cannot occupy the same point in space)
\begin{equation*}
  \pp : \body \to \espace,
\end{equation*}
sometimes referred to as a \emph{placement} in mechanics and its image $\Omega_{\pp} = \pp(\body)$ is usually called a \emph{configuration system}. The \emph{configuration space} in continuum mechanics is thus the set, noted $\Emb(\body,\espace)$, of smooth orientation-preserving embeddings of $\body$ in $\espace$ (in this paper, all the embeddings are assumed to be smooth, in \cite{Seg1986,SE2020}, they are assumed to be $C^p$, $p \geq 1$). This set can be endowed with a differential manifold structure of infinite dimension~\cite{IKT2013} (a \emph{Fréchet} manifold). Since $\espace$ is an affine space, $\Emb(\body,\espace)$ can be considered as an open set of a \emph{Fréchet topological vector space} (see~\autoref{sec:Frechet-topology}).

The tangent space to $\Emb(\body,\espace)$ at a point $\pp \in \Emb(\body,\espace)$ is described as follows. Let $\pp(t)$ be a smooth curve in $\Emb(\body,\espace)$ such that $\pp(0) = \pp$ and $(\partial_{t}p)(0) = \VV$, then the tangent space at $\pp \in \Emb(\body,\espace)$ is the set
\begin{equation*}
  T_{\pp}\Emb(\body,\espace) = \set{\VV \in \Cinf(\body,T\espace);\; \pi \circ \VV = \pp},
\end{equation*}
where $\VV$ is described by the following diagram:
\begin{equation*}
  \xymatrix{
    T\body \ar[r]^{T\pp} \ar[d]_{\pi}   & T\espace \ar[d]^{\pi} \\
    \body \ar[ur]^{\VV} \ar[r]^{\pp}    & \espace }
\end{equation*}
We recognize $\VV$ as a \emph{Lagrangian velocity}. The tangent bundle $T\Emb(\body,\espace)$ of the configuration space $\Emb(\body,\espace)$ is thus the set of \emph{virtual Lagrangian velocities}.

\begin{rem}\label{TEmb-trivial}
  Since $T\espace=\espace \times E$ is trivial, the tangent bundle $T\Emb(\body,\espace)$ is trivial and isomorphic to $\Emb(\body,\espace) \times C^\infty(\body, E)$ (see~\cite[p. 107]{EM1970a}).
\end{rem}

A motion in continuum mechanics corresponds to a smooth curve $\pp(t)$ in $\Emb(\body,\espace)$ (a \emph{path of embeddings}). To this motion is associated its Lagrangian velocity
\begin{equation*}
  \partial_{t}p(t,\XX) = \VV(t,\XX)
\end{equation*}
and its (right) Eulerian velocity
\begin{equation*}
  \uu(t,\xx) = \VV(t,\pp^{-1}(t,\xx)).
\end{equation*}
which is, at each time $t$ a vector field on $\Omega_{\pp(t)} = \pp(t)(\body)$, simply noted $\Omega_{\pp} = \pp(\body)$. But one can also introduces the \emph{left Eulerian velocity}
\begin{equation*}
  \UU(t,\XX) := (T\pp^{-1}. \VV)(t,\XX).
\end{equation*}
which is a vector field on $\body$ (see for instance~\cite{Arn1966}). These two vector fields are better described by the following diagram:
\begin{equation*}
  \xymatrix{
  T\body \ar[r]^{T\pp} \ar[d]_{\pi}   & T\Omega_{\pp} \ar[d]^{\pi}
  \\
  \body \ar@/^1pc/[u]^{\UU} \ar[ur]^{\VV} \ar[r]^{\pp}    & \Omega_{\pp}  \ar@/_1pc/[u]_{\uu}
  }
\end{equation*}

\begin{rem}
  Introducing the notions of \emph{pullback}/\emph{pushforward} (see~\autoref{sec:pullback}), left and right Eulerian velocities are related to each other as follows,
  \begin{align*}
    \UU & = \pp^{*}\uu = T\pp^{-1} \circ \uu \circ \pp, &  & \text{($\UU$ is the pullback of $\uu$)},
    \\
    \uu & = \pp_{*}\UU = T\pp \circ \UU \circ \pp^{-1}, &  & \text{($\uu$ is the pushforward of $\UU$)}.
  \end{align*}
\end{rem}

It is common to introduce a \emph{reference configuration} $\Omega_{0} := \pp_{0}(\body)$ with $\xx_{0}=\pp_{0}(\XX)$. This allows for the definition of the mapping
\begin{equation*}
  \varphi := \pp \circ \pp_{0}^{-1}, \qquad \Omega_{0} \to \Omega_{\pp},
\end{equation*}
usually called the \emph{deformation}~\cite{Eri1962,GZ1968,TN1965,WT1973}. This means that $\xx=\pp(\XX)=\varphi(\xx_{0})$ are the Eulerian coordinates (on the deformed configuration $\Omega_{\pp}$), whereas $\xx_{0}$ and $\XX$ are the Lagrangian coordinates, respectively on the reference configuration $\Omega_{0}$ and the body $\body$.

The linear tangent mappings $T\pp_{0} : T\body \to T\Omega_{0}$, $T\pp : T\body \to T\Omega_{\pp}$ and $T\varphi : T\Omega_{0} \to T\Omega_{\pp}$, will be denoted respectively by $\bF_{0} := T\pp_{0}$, $\bF := T\pp$, and $\bF_{\varphi} := T\varphi$. In local coordinate systems, these linear mappings are represented respectively by the following matrices
\begin{equation*}
  \left(\frac{\partial {\pp_{0}}^{i}}{\partial X^{J}} \right),
  \qquad
  \left(\frac{\partial \pp^{i}}{\partial {X}^{J}} \right),
  \qquad
  \left(\frac{\partial \varphi^{i}}{\partial {x_{0}}^{j}} \right).
\end{equation*}

The \emph{displacement}, either defined on the reference configuration $\Omega_{0}=\pp_{0}(\body)$ as
\begin{equation}\label{eq:xi0}
  \bxi(\varphi)= \varphi - \id,
\end{equation}
or on the body $\body$, as
\begin{equation}\label{eq:xiB}
  \bxi(\pp) = ( \varphi - \id) \circ \pp_{0} =\pp - \pp_{0},
\end{equation}
can be considered as a \emph{vector valued function}, with values in $E$, since $\pp(\XX)-\pp_{0}(\XX)$ belongs to $E$, the vector space associated to $\espace$. More generally, an embedding $\pp: \body \to \espace$ can be considered as a vector valued function $\pp: \body \to E$, as soon as we have fixed an origin in space (this choice is not necessary for the displacement which is canonically a vector valued function). Here, we strongly use the affine structure of space $\espace$ which simplifies the expression of the displacement which is, by the way, a vector field on $\Emb(\body,\espace)$, as justified in~\cite{GLM2014} where the rigorous formulation of the displacement requires, then, the non-trivial use of the \emph{Riemannian exponential mapping} on $\espace$ (considered as a general Riemannian manifold with no affine structure).

\subsection{Metric states}
\label{subsec:metric-states}

To each embedding $\pp$ corresponds, by pullback of the Euclidean metric $\bq$, a Riemannian metric
\begin{equation*}
  \bgamma = \pp^{*}\bq
\end{equation*}
on the body $\body$. Note that the Riemannian curvature of $\bgamma$ vanishes when $\body$ is of dimension 3 (this is no longer true in shell theory, where $\body$ is a manifold of dimension 2). Indeed, the Riemannian curvature tensor $R$ is covariant to the metric, meaning that if $\phi$ is a diffeomorphism between two manifolds $\body$ and $\Omega$, and $g$ is a metric on $\Omega$, then, $R(\phi^{*}g) = \phi^{*}R(g)$. In particular, for $\phi=\pp : \body \to \Omega_{\pp}$, we get $R(\bgamma) = R(\pp^{*}\bq) = \pp^{*}R(\bq) = 0$.

For both fluids and solids, the state of stress is determined by the ``metric state''~\cite{TN1965,Nol1972}. It is however \emph{necessary to introduce a reference configuration $\pp_{0}$ to formalize geometrically hyper-elasticity of solids}, which leads to a reference Riemannian metric
\begin{equation*}
  \bgamma_{0} = {\pp_{0}}^{*}\bq.
\end{equation*}

In classical theory of finite strains, two strain tensors play a fundamental role:
\begin{enumerate}
  \item the (covariant) \emph{right Cauchy--Green tensor}
        \begin{equation*}
          \bC := \varphi^{*}\bq = \bq\bF_{\varphi}^{t}\bF_{\varphi},
        \end{equation*}
        defined as the \emph{pullback} to $\Omega_{0}$ of the Euclidean metric $\bq$ on $\Omega_{\pp}$,
  \item the (contravariant) \emph{left Cauchy--Green tensor}
        \begin{equation*}
          \bb := \varphi_{*} \bq^{-1} = \bF_{\varphi}\bF_{\varphi}^{t}\bq^{-1},
        \end{equation*}
        defined as the \emph{pushforward} to $\Omega_{\pp}$ of the inverse Euclidean metric $\bq^{-1}$ on $\Omega_{0}$,
\end{enumerate}
where $(\cdot)^{t}$ denotes the transpose relative to the metric $\bq$.

The relations between the two Cauchy--Green tensors and the two metric tensors $\bgamma$ and $\bgamma_{0}$, defined on the body $\body$ are the following
\begin{equation*}
  {\pp_{0}}^{*}\bC = \pp^{*}\bq := \bgamma , \quad \text{and} \quad p^{*}\bb = {\pp_{0}}^{*}\bq^{-1} := \bgamma_{0}^{-1}.
\end{equation*}

\begin{rem}\label{rem:mixedCb}
  Using the reference metric tensor $\bgamma_{0}$, we can then build the mixed tensor $\bgamma_{0}^{-1}\bgamma$. Then, $\widehat \bC=\bq^{-1}\bC$ and $\widehat \bb=\bb\bq$ are related to each other by
  \begin{equation}\label{eq:mixedCb}
    \bgamma_{0}^{-1}\bgamma=\pp_{0}^{*}\, \widehat \bC=\pp^{*}\, \widehat \bb.
  \end{equation}
\end{rem}

\subsection{Strain rate}
\label{subsec:strain-rate}

Traditionally, the \emph{strain rate} (on $\Omega_{\pp}$) is defined by the mixed tensor
\begin{equation*}
  \widehat{\bd} := \frac{1}{2} \left( \nabla \uu + (\nabla \uu)^{t}\right),
\end{equation*}
where $\nabla \uu$ is the covariant derivative of the Eulerian velocity $\uu$ and $(\nabla \uu)^{t}$
means the transpose (relative to the metric $\bq$) of the linear operator $\ww \mapsto \nabla_{\ww} \uu$. It appears, however, that it is more interesting to introduce its covariant form
\begin{equation*}
  \bd := \bq \widehat{\bd} = \frac{1}{2} \Lie_{\uu} \bq,
\end{equation*}
where $\Lie_{\uu}$ means the Lie derivative with respect to $\uu$. We get then
\begin{equation}\label{eq:pbd}
  \varphi^{*}\bd = \frac{1}{2} \partial_{t} \bC, \qquad \pp^{*}\bd = \frac{1}{2} \partial_{t} \bgamma,
\end{equation}
which are direct consequences of lemma~\ref{lem:magic-formula} (see also~\cite{Nol1972,Rou1980,Rou2006}).

\subsection{Mass measure}
\label{subsec:mass-measure}

Let $\mu$ be the mass measure on the body $\body$, that we assume to be given by a volume form, \textit{i.e.}, a nowhere vanishing 3-form, and not as a general Borel measure as in~\cite{Nol1978}. To each embedding $\pp$ corresponds thus a volume form $\pp_{*}\mu$ on the configuration $\Omega_{\pp} = \pp(\body)$. Besides, the Riemannian metric $\bq$ induces on $\Omega_{\pp}$ a volume form, noted $\vol_{\bq}$ (see~\autoref{sec:volume-forms}), which can be written as
\begin{equation*}
  \vol_{\bq}=dx^{1} \wedge dx^{2} \wedge dx^{3},
\end{equation*}
in each direct orthonormal coordinate system $(x^{i})$ of the oriented affine Euclidean space $\espace$.
Since two volume forms are always proportional to each other, we get
\begin{equation}\label{eq:mass-density}
  \pp_{*}\mu = \rho \, \vol_{\bq},
\end{equation}
which defines the \emph{mass density} $\rho=\pp_{*}\mu / \vol_{\bq}$ as a scalar function on the space domain $\Omega_{\pp}$.
Similarly, a density $\rho_{0}=\pp_{0 *}\mu / \vol_{\bq}$ is defined on the reference configuration $\Omega_{0}$. The conservation of mass is obtained using the identity
\begin{equation*}
  \mu = {\pp_{0}}^{*}(\rho_{0} \, \vol_{\bq}) = \pp^{*}(\rho \, \vol_{\bq}),
\end{equation*}
from which we deduce,
\begin{equation}\label{eq:mass-conservation}
  \rho_{0} \, \vol_{\bq} = {\pp_{0}}_{*}\pp^{*}(\rho \, \vol_{\bq}) = \varphi^{*}(\rho \, \vol_{\bq}) = (\varphi^{*}\rho)J_{\varphi}\, \vol_{\bq},
  \qquad
  \begin{cases}
    J_{\varphi}=\det \bF_\varphi,
    \\
    (\varphi^{*}\rho)=\rho\circ \varphi.
  \end{cases}
\end{equation}
Its infinitesimal form
\begin{equation*}
  \partial_{t} \rho + \nabla_{\uu} \rho + \rho \dive \uu = 0
\end{equation*}
is obtained by deriving along a path of embeddings $\pp(t)$ and using lemma~\ref{lem:magic-formula}. Indeed, we get
\begin{equation*}
  \partial_{t}\mu = 0 = \partial_{t}\left(\pp^{*}(\rho \vol_{\bq}) \right) = \pp^{*} \left( \partial_{t}(\rho \vol_{\bq}) + \Lie_{\uu}(\rho\vol_{\bq}) \right),
\end{equation*}
where $\Lie_{\uu}$ is the Lie derivative relative to $\uu$, and from which we deduce that
\begin{equation*}
  (\partial_{t} \rho) \vol_{\bq} + (\Lie_{\uu}\rho) \vol_{\bq} + \rho \Lie_{\uu}\vol_{\bq} = \left(\partial_{t}\rho + \nabla_{\uu} \rho + \rho \dive \uu \right)  \vol_{\bq} = 0.
\end{equation*}

The Riemannian volume forms $\vol_{\bgamma}$ and $\vol_{\bgamma_{0}}$ and the mass measure $\mu$ are all proportional. They are related by
\begin{equation*}
  \mu = (\pp^{*}\rho) \vol_{\bgamma} = (\pp_{0}^{*}\, \rho_{0}) \vol_{\bgamma_{0}},
\end{equation*}
which leads us to define two mass densities on the body
\begin{equation*}
  \rho_\bgamma := \pp^{*}\rho=\rho \circ \pp, \qquad  \rho_{\bgamma 0} := \pp_{0}^{*}\, \rho_{0} = \rho_{0} \circ \pp_{0}.
\end{equation*}

\section{The Rougée geometric formulation of hyper-elasticity revisited}
\label{sec:hyper-elasticity}

Rougée suggested in~\cite{Rou1997,Rou2006} to define an elasticity law as a vector field on $\Met(\body)$, the manifold of all the Riemannian metrics on the body $\body$, in other words as a section
\begin{equation*}
  F: \Met(\body) \to T\Met(\body)
\end{equation*}
of the tangent vector bundle $T\Met(\body)$. It turns out that $\Met(\body)$ is an \emph{open convex set} of the infinite dimensional vector space
\begin{equation*}
  \Gamma(S^{2}T^{\star}\body),
\end{equation*}
of smooth covariant symmetric second-order tensors fields (see~\autoref{sec:Frechet-topology}). The tangent space $T_{\bgamma}\Met(\body)$ is thus canonically identified with the vector space $\Gamma(S^{2}T^{\star}\body)$. This space can be interpreted as the space of virtual deformation tensor fields $\bepsilon$ (linearized deformations around metric state $\bgamma$). The tangent vector bundle $T\Met(\body)$ is trivial and can be written as
\begin{equation*}
  T\Met(\body) = \Met(\body) \times \Gamma(S^{2}T^{\star}\body).
\end{equation*}

The cotangent vector space $T_{\bgamma}^{\star}\Met(\body)$ is here a space of \emph{tensor-distributions}, as introduced by Lichnerowicz in~\cite{Lic1994}. Tensor-distributions extend \emph{Schwartz distributions} from functions to tensor fields. Hence, $T_{\bgamma}^{\star}\Met(\body)$ is the topological dual of $\Gamma(S^{2}T^{\star}\body)$ (for the Fréchet topology defined by the $C^{k}$ semi-norms, see~\autoref{sec:Frechet-topology})). $T_{\bgamma}^{\star}\Met(\body)$ is thus interpreted as the space of \emph{virtual powers of internal forces} in a very general sense.

\begin{rem}
  Recall that each pullback metric $\gamma=\pp^{*}\bq$ has vanishing curvature (see~\autoref{subsec:metric-states}), so that the set of all pullback metrics is a strict subset of $\Met(\body)$.
\end{rem}

The manifold $\Met(\body)$ of Riemannian metrics on $\body$ can be equipped with a (weak) Riemannian structure, by setting
\begin{equation}\label{eq:Rougee-metric}
  G^{\mu}_{\bgamma}(\bepsilon^{1}, \bepsilon^{2}) := \int_{\body} \tr(\bgamma^{-1}\bepsilon^{1} \bgamma^{-1}\bepsilon^{2}) \, \mu, \qquad \bepsilon^{1}, \bepsilon^{2} \in T_{\bgamma}\Met(\body),
\end{equation}
where $\tr(\bgamma^{-1}\bepsilon^{1} \bgamma^{-1}\bepsilon^{2}) = \gamma^{ij}\varepsilon^{1}_{ik} \gamma^{kl}\varepsilon^{2}_{jl}$, in a local coordinate system (note however that the trace $\tr(\bgamma^{-1}\bepsilon^{1} \bgamma^{-1}\bepsilon^{2})$ is intrinsic).

Riemannian structures on the manifold of Riemannian metrics have been extensively studied, see for instance~\cite{Ebi1968,FG1989,GM1991,Cla2010}. The metric~\eqref{eq:Rougee-metric} was introduced by Rougée~\cite{Rou1997,Rou2006} (see also \cite{Fia2004}) and seems well adapted for the geometrical formulation of several concepts in solid mechanics. For this reason, we will call it the \emph{Rougée metric}. This metric induces a linear injective (but not surjective) mapping
\begin{equation*}
  T_{\bgamma}\Met(\body) \to T_{\bgamma}^{\star}\Met(\body), \qquad \bbeta \mapsto G^{\mu}_{\bgamma}(\bbeta,\cdot).
\end{equation*}

The range of this mapping in $T_{\bgamma}^{\star}\Met(\body)$ corresponds to \emph{virtual powers with density}. In other words, an element $\mathcal{P}_{\bgamma}$ belongs to this range if it can be written as
\begin{equation*}
  \mathcal{P}_{\bgamma}(\bepsilon) = \int_{\body} (\btheta : \bepsilon) \mu, \quad \text{where} \quad \btheta = \bgamma^{-1} \bbeta \bgamma^{-1},
\end{equation*}
for some $\bbeta \in T_{\bgamma}\Met(\body)$, defining on the body the \emph{Rougée stress tensor} $\btheta$ \cite{Rou1980,Rou1991} as the density of the power $\mathcal{P}_{\bgamma}$. An elasticity law (in the Cauchy sense) is thus
\begin{equation}\label{eq:elastic-law-in-the-body}
  \btheta = \bgamma^{-1} F(\bgamma) \bgamma^{-1},
\end{equation}
where $F$ is a vector field on $\Met(\body)$. This formula is better understood using the following diagram
\begin{equation*}
  \xymatrix{
  T\Met(\body) \ar[r]^{G^{\mu}} \ar[d]^{\pi} & T^{\star}\Met(\body) \\
  \Met(\body) \ar@/^1pc/[u]^{F} \ar[ur]_{\quad \btheta = \bgamma^{-1} F(\bgamma) \bgamma^{-1}} &  }
\end{equation*}
The \emph{Noll intrinsic stress tensor}~\cite{Nol1972}, also defined on the body $\body$, is
\begin{equation*}
  \msigma := \rho_\bgamma\, \btheta, \qquad
  \rho_\bgamma = p^{*} \rho.
\end{equation*}
It is such that
\begin{equation*}
  \mathcal{P}_{\bgamma}(\bepsilon) = \int_{\body} (\msigma : \bepsilon) \vol_\bgamma,
  \qquad  \bepsilon\in T_{\bgamma}\Met(\body).
\end{equation*}
Both $\msigma$ and $\btheta$ are symmetric contravariant tensor fields on $\body$.

By \emph{pushforward} on $\Omega_{\pp}=\pp(\body)$, we recover then the Kirchhoff and Cauchy stress tensors
\begin{equation}\label{eq:elastic-law-in-space}
  \btau = \pp_{*}\btheta = \, \bq^{-1} \pp_{*}(F(\bgamma))\bq^{-1},
  \qquad
  \bsigma=\rho \btau= \pp_{*}\msigma.
\end{equation}

\begin{rem}
  Using the expression of the deformation rate~\eqref{eq:pbd}, the evaluation
  \begin{equation*}
    \mathcal{P}_{\bgamma}\Big(\frac{1}{2} \partial_{t} \bgamma\Big)=
    \int_{\pp(\body)} \pp_{*} \btheta \2dots \pp_{*}\Big(\frac{1}{2}  \partial_{t} \bgamma \Big)\, \pp_{*}\mu
    =\int_{\Omega_{\pp}} \rho \btau \2dots \bd\, \vol_{\bq}
    =\int_{\Omega_{\pp}} \bsigma \2dots \bd\, \vol_{\bq}
    =- \mathcal{P}^{int}
  \end{equation*}
  is recognized as the opposite of the power of internal forces for a Cauchy medium.
\end{rem}

In this geometric framework, the law~\eqref{eq:elastic-law-in-the-body} is said to be \emph{hyper-elastic} (or elastic in Green's sense) if $F$ is the \emph{gradient} (for the Rougée metric $G^{\mu}$) of a functional $\mathcal H$ defined on $\Met(\body)$. In that case
\begin{equation*}
  F = \grad^{G^{\mu}}\mathcal H,
\end{equation*}
where $\grad^{G^{\mu}}\mathcal H$ is defined implicitly by
\begin{equation*}
  d_{\bgamma}\mathcal H.\bepsilon = G^{\mu}_{\bgamma}(\grad^{G^{\mu}}\mathcal H, \bepsilon) = \int_{\body} \tr (\bgamma^{-1}(\grad^{G^{\mu}}\mathcal H) \bgamma^{-1} \bepsilon)\,\mu, \qquad \forall \bepsilon \in T_{\bgamma}\Met(\body),
\end{equation*}
meaning that
\begin{equation*}
  \mathcal{P}_{\bgamma}(\bepsilon) = \int_{\body} (\btheta : \bepsilon) \mu=d_{\bgamma}\mathcal H.\bepsilon.
\end{equation*}

\begin{rem}
  There is a strong similarity between this geometric framework and Einstein--Hilbert formulation of general relativity~\cite{Hil1924,Ein1988,Sou1964}. In the second case, the metric $\bgamma$ on the body is replaced by a Lorentzian metric $g$ on space-time $\mathcal{U}$, the Rougée stress tensor $\btheta$ becomes the \emph{stress–energy tensor} and the vector field $F(\bgamma)= \grad^{G^{\mu}} \mathcal{H}$ is replaced by the Einstein tensor
  \begin{equation*}
    S(g) = \grad^{G^{E}}\mathcal H = a \mathrm{Ric}(g) - \frac{1}{2}(aR(g) + b)g,
  \end{equation*}
  which is the gradient of the Einstein--Hilbert functional
  \begin{equation*}
    \mathcal H(g)=\int_{\mathcal U} (aR(g) +b) \vol_{g},
  \end{equation*}
  for the \emph{Ebin metric}~\cite{Ebi1968}
  \begin{equation*}
    G^{E}_{g}(\bepsilon^{1}, \bepsilon^{2}) := \int_{\mathcal{U}} \tr(g^{-1}\bepsilon^{1} g^{-1}\bepsilon^{2}) \, \vol_{g}.
  \end{equation*}
  The two constants $a$ and $b$ are related to the Newtonian constant of gravitation and to the Einstein cosmological constant \cite[p. 344, Eq. (35.25)]{Sou1964}.
\end{rem}

\subsection{Isotropic hyper-elasticity formulated on the body}
\label{subsec:isotropic-hyper-elasticity}

One may refer to~\cite{SM1984,SH1997} and~\cite[Chapter 7]{Ber2012} for a review of usual hyper-elasticity laws
and their specific free energies $\psi$. Here, we choose to formulate hyper-elasticity on the body $\body$ rather than on a reference configuration $\Omega_{0}$, using the metric $\bgamma=\pp^{*} \bq$ and the Rougée stress $\btheta=\pp^{*} \btau$. In that case, a reference metric $\bgamma_{0}= \pp_{0}^{*}\bq$ is required in order to define the mixed tensor~\eqref{eq:mixedCb} and its invariants. The formulation of isotropic hyper-elasticity (in Green sense) on the body $\body$, is provided by the following result.

\begin{thm}\label{thm:iso-hyper-body}
  Let $\pp_{0}$ be a reference configuration and $\bgamma_{0}=\pp_{0}^{*}\,\bq$. The isotropic local hyper-elasticity can be formulated on the body $\body$ by a functional
  \begin{equation*}
    \mathcal H_{\bgamma_{0}}(\bgamma) = \int_{\body} 2\psi (I_{1},  I_{2}, I_{3}) \mu
  \end{equation*}
  defined on $\Met(\body)$, where $I_{k}=\tr(\bgamma_{0}^{-1} \bgamma)^{k}$.
  We get
  \begin{equation}\label{eq:hyper-ela-theta}
    \btheta = \sum_{k=1}^{3} 2  \left(\frac{\partial \psi}{\partial I_{k}}(\bgamma)\right)\frac{\partial I_k}{\partial \bgamma}
    = \sum_{k=1}^{3} 2  k \left(\frac{\partial \psi}{\partial I_{k}}(\bgamma)\right)(\bgamma_{0}^{-1} \bgamma)^{k-1}\bgamma_{0}^{-1}.
  \end{equation}
\end{thm}

\begin{rem}\label{rem:I3}
  In the argument of $\psi$, $I_{3}$ is often replaced by $\mathcal{J} = \sqrt{\det(\bgamma_{0}^{-1}\bgamma)}$
  (whose pushforward on $\Omega_{0}$ is  $\sqrt{\det\widehat \bC}=J_{\varphi}$), with then
  \begin{equation*}
    \frac{\partial \mathcal{J}}{\partial \bgamma} = \frac{1}{2} \mathcal{J}\, \bgamma^{-1}.
  \end{equation*}
\end{rem}

\begin{proof}
  We have
  \begin{equation*}
    \delta \mathcal H_{\bgamma_{0}} = \int_{\body} 2 \sum_{k=1}^{3} \frac{\partial \psi}{\partial I_{k}} \delta I_{k} \, \mu,
  \end{equation*}
  where
  \begin{equation*}
    \delta I_{k} = \delta (\tr(\bgamma_{0}^{-1} \bgamma)^{k}) = k \tr \left[(\bgamma_{0}^{-1} \bgamma)^{k-1}\bgamma_{0}^{-1}\delta \bgamma\right] = k \tr \left[\bgamma^{-1}(\bgamma(\bgamma_{0}^{-1} \bgamma)^{k})\bgamma^{-1}\delta \bgamma\right].
  \end{equation*}
  We get thus
  \begin{align*}
    \delta \mathcal H_{\bgamma_{0}} & = \int_{\body} \left(\sum_{k=1}^{3} 2k \frac{\partial \psi}{\partial I_{k}}(\bgamma) \tr \left[\bgamma^{-1}(\bgamma(\bgamma_{0}^{-1} \bgamma)^{k})\bgamma^{-1}\delta \bgamma\right] \right) \mu
    \\
                                    & = \int_{\body} \tr \left[\bgamma^{-1}\left(\sum_{k=1}^{3} 2 k \left(\frac{\partial \psi}{\partial I_{k}}(\bgamma)\right) \bgamma(\bgamma_{0}^{-1} \bgamma)^{k}\right)\bgamma^{-1}\delta \bgamma\right] \, \mu.
  \end{align*}
  Therefore
  \begin{equation*}
    \grad^{G^{\mu}}\mathcal H_{\bgamma_{0}} = \sum_{k=1}^{3} 2k \left(\frac{\partial \psi}{\partial I_{k}}(\bgamma)\right) \bgamma(\bgamma_{0}^{-1} \bgamma)^{k}
  \end{equation*}
  and
  \begin{equation*}
    \btheta = \bgamma^{-1} \left(\grad^{G^{\mu}}\mathcal H_{\bgamma_{0}}\right) \bgamma^{-1} = \sum_{k=1}^{3} 2k \left(\frac{\partial \psi}{\partial I_{k}}(\bgamma)\right)(\bgamma_{0}^{-1} \bgamma)^{k-1}\bgamma_{0}^{-1}
    = \sum_{k=1}^{3} 2 \frac{\partial \psi}{\partial I_k} \frac{\partial I_k}{\partial \bgamma},
  \end{equation*}
  which completes the proof.
\end{proof}

The isotropic invariants of the mixed tensor $\bgamma_{0}^{-1} \bgamma$ can be rewritten using either the left Cauchy--Green tensor $\bb = \pp_{*}\bgamma_{0}^{-1}$ (on $\Omega_{\pp}$) or the right Cauchy--Green tensor $\bC = {\pp_{0}}_{*} \bgamma$  (on $\Omega_{0}$), using the identities in remark~\ref{rem:mixedCb}. Indeed, as
pushforwards commute with all tensor contractions, we have
\begin{equation*}
  \pp_{*}I_{k} = I_{k} \circ \pp^{-1} = \tr (\bb\bq)^{k} = \tr \widehat \bb^{k},
\end{equation*}
and
\begin{equation*}
  {\pp_{0}}_{*}I_{k} = I_{k} \circ \pp_{0}^{-1}  = \tr (\bq^{-1}\bC)^{k} = \tr \widehat \bC^{k},
\end{equation*}
from which we deduce that
\begin{equation*}
  \frac{\partial \pp_{*}I_{k}}{\partial \bb} = k \bq (\bb\bq)^{k-1}, \qquad \frac{\partial {\pp_{0}}_{*}I_{k}}{\partial \bC} = k (\bq^{-1}\bC)^{k-1}\bq^{-1}.
\end{equation*}

Consider now the second Piola--Kirchhoff stress tensor
\begin{equation}\label{eq:PK2}
  \bS: = \varphi^{*} \btau = \pp_{0*} \btheta,
\end{equation}
defined as the pullback on $\Omega_{0}$, by the deformation $\varphi=\pp\circ \pp_{0}^{-1}$, of Kirchhoff stress tensor $\btau$. Then, we recover, on the one hand, the usual hyper-elasticity law
\begin{equation}\label{eq:HyperSC}
  \bS = 2\frac{\partial \psi_{0}}{\partial \bC},
  \qquad
  \psi_{0}(\widehat\bC) = \psi(I_{1} \circ \pp_{0}^{-1},I_{2} \circ \pp_{0}^{-1},I_{3} \circ \pp_{0}^{-1}),
\end{equation}
expressed on the reference configuration $\Omega_{0}$ by \emph{pushforward} of~\eqref{eq:hyper-ela-theta} as
\begin{equation*}
  \bS = {\pp_{0}}_{*} \btheta = 2k \left(\frac{\partial \psi}{\partial I_{k}}(\bgamma) \circ \pp_{0}\right)(\bq^{-1} \bC)^{k-1}\bq^{-1} = 2\frac{\partial \psi_{0}}{\partial \bC},
\end{equation*}
and, on the other hand, the usual isotropic hyper-elasticity law
\begin{equation}\label{eq:HyperK}
  \btau = 2\bq^{-1}\frac{\partial \psi_{\pp}}{\partial \bb}\bb,
  \qquad
  \psi_{\pp}(\widehat\bb)=\psi(I_{1} \circ \pp^{-1},I_{2} \circ \pp^{-1},I_{3} \circ \pp^{-1}),
\end{equation}
expressed on the deformed configuration $\Omega_{\pp}$ by \emph{pushforward} of~\eqref{eq:hyper-ela-theta} as
\begin{equation*}
  \btau = \pp_{*}\btheta =  2 k \left(\frac{\partial \psi}{\partial I_{k}}(\bgamma)\circ \pp\right) (\bb\bq)^{k-1}\bb
  = 2\bq^{-1}\frac{\partial \psi_{\pp}}{\partial \bb}\bb.
\end{equation*}

\begin{rem}
  Note that identifying the body with the reference configuration $\body \equiv \Omega_{0}$, and thus $\pp_{0}\equiv \id$, $\pp \equiv \varphi$, leads to the identification
  \begin{equation*}
    \btheta \equiv \bS,
  \end{equation*}
  of the Rougée tensor with the second Piola--Kirchhoff tensor.
\end{rem}

\subsection{Example: Harth--Smith hyper-elasticity}
\label{subsec:Harth-Smith}

As an example, let us consider \emph{quasi-incompressible Harth--Smith hyper-elasticity}~\cite{Hart66}, formulated on the reference configuration $\Omega_{0}$ by \eqref{eq:HyperSC} with the isotropic specific free energy
\begin{equation*}
  \psi_{0}= h_{1}\int_{3}^{\overline{I}_{1}(\widehat \bC)}{ \exp{(h_{3}(I-3)^2)}d I }+ 3h_{2}
  \ln\left(\frac{\overline{I}_{2}(\bC)}{3}\right) +U(J_\varphi),
  \qquad
\end{equation*}
where $h_{1}$, $h_{2}$ and $h_{3}$ are material parameters, $\widehat \bC=\bq^{-1} \bC$ and $J_\varphi=\sqrt{\det \widehat \bC}$. In order to properly model rubber quasi-incompressibility, the Penn invariants~\cite{Pen1970}  $\overline{I}_{1}$, $\overline{I}_{2}$ are used for the shear part of $\psi_{0}$, defined on $\Omega_{0}$ as
\begin{equation*}
  \overline{I}_{1}(\widehat \bC) :=(\det \widehat \bC)^{-\frac{1}{3}} \tr\widehat \bC,
  \qquad
  \overline{I}_{2}(\bC) :=(\det \widehat \bC)^{-\frac{2}{3}} \tr(\widehat \bC)^2,
\end{equation*}
together with the compressibility function
\begin{equation*}
  U(J)= \kappa (J\ln J-J+1),
\end{equation*}
defining $K=\rho_{0} \kappa$ as the compressibility, a material parameter~\cite{DS1999}. Since $\pp_{0}^{*}\, \widehat \bC=\bgamma_{0}^{-1}\bgamma$, we get
\begin{equation*}
  \pp_{0}^{*} \, J_\varphi = \mathcal{J},
  \qquad
  \pp_{0}^{*} \, \overline{I}_{1}(\widehat \bC) = \mathcal{J}^{-\frac{1}{3}} I_{1},
  \qquad
  \pp_{0}^{*} \,\overline{I}_{2}(\widehat\bC) = \mathcal{J}^{-\frac{2}{3}} I_{2}.
\end{equation*}
where $\mathcal{J} = (\det(\bgamma_{0}^{-1}\bgamma))^{1/2}$, $I_{1}=\tr(\bgamma_{0}^{-1}\bgamma)$ and $I_{2}=\tr(\bgamma_{0}^{-1}\bgamma\, \bgamma_{0}^{-1}\bgamma)$ are functions on $\body$. We deduce then, by theorem~\ref{thm:iso-hyper-body} and remark~\ref{rem:I3}, that Hart--Smith hyper-elasticity can be formulated on the body $\body$, as~\eqref{eq:hyper-ela-theta}, with
\begin{equation*}
  \psi=\pp_{0}^{*}\, \psi_{0} = h_{1}\int_{3}^{\mathcal{J}^{-\frac{1}{3}} I_{1}}{ \exp{(h_{3}(I-3)^2)}d I } + 3h_{2}
  \ln\left(\frac{\mathcal{J}^{-\frac{2}{3}} I_{2}}{3}\right) + U(\mathcal{J}).
\end{equation*}

\section{The virtual power as a one-form on the configuration space}
\label{sec:virtual-power}

In this section, we fix our notations in order to recast the virtual power of applied surface forces as a 1-form on the configuration space and formulate an elasticity problem in a geometric manner, directly on the body $\body$. We start by recalling that the principle of virtual power in quasi-statics states that
\begin{equation*}
  \mathcal{P}^{int}(\ww) + \mathcal{P}^{ext}(\ww) = 0,
\end{equation*}
for any kinematically admissible virtual displacement field $\ww$ on $\Omega_{\pp}$, sufficiently regular and vanishing on the prescribed displacement (Dirichlet) boundary $\partial \Omega_{\pp}^{\bxi} \subset \partial \Omega_{\pp}$. In the sequel, we assume that $\partial \Omega_{\pp}^{\bxi} \neq \emptyset$, in order ton avoid some indeterminacy.

The virtual power of the internal forces $\mathcal{P}^{int}$ is
\begin{equation*}
  - \int_{\Omega_{\pp}} (\bsigma : \bepsilon) \, \vol_{\bq},
\end{equation*}
where $\bepsilon$ is a virtual strain field (a second order symmetric covariant tensor field) on $\Omega_{\pp}$. It is further assumed that the virtual deformation field $\bepsilon$ derives from a virtual displacement field $\ww$ on $\Omega_{\pp}$, which means that
\begin{equation*}
  \bepsilon := \frac{1}{2}\Lie_{\ww}\bq = D\ww^{\flat},
  \qquad
  \ww^{\flat}=\bq \ww.
\end{equation*}
Here, the operator
\begin{equation}\label{eq:DZ}
  D\ww^{\flat} (X,Y) = \frac{1}{2} \left( (\nabla_X \ww^{\flat})(Y) + (\nabla_Y \ww^{\flat})(X)\right)
\end{equation}
is the \emph{formal adjoint} of the divergence operator. One, then, has
\begin{equation*}
  \mathcal{P}^{int}(\ww) = - \int_{\Omega_{\pp}} (\bsigma : D\ww^{\flat}) \, \vol_{\bq} = -\int_{\partial \Omega_{\pp}} (\widehat{\bsigma}\nn \cdot \ww) da + \int_{\Omega_{\pp}} (\dive \bsigma \cdot \ww)\,\vol_{\bq},
\end{equation*}
where $\nn$ is the unit outer normal on the boundary $\partial \Omega_{\pp}$, $da=i_{\vec n}\vol_{\bq}$ is the area element, and $\widehat{\bsigma}:=\bsigma \bq$ is the mixed form of Cauchy stress tensor.

The virtual power of external forces has for general expression
\begin{equation}\label{eq:ext-power}
  \mathcal{P}^{ext} = \mathcal{P}^{ext,v} + \mathcal{P}^{ext,s},
\end{equation}
where
\begin{equation*}
  \mathcal{P}^{ext,v}(\ww) = \int_{\Omega_{\pp}} \left(\ff_{v}\cdot \ww \right) \vol_{\bq},
\end{equation*}
and $\ff_{v}$ is the field of external forces per unit of volume. The second term $\mathcal{P}^{ext,s}$ corresponds to surface forces (boundary conditions of Neumann type). In the sequel, we will assume that $\mathcal{P}^{ext,s}$ consists in a so-called \emph{dead load} term (\emph{DL}) and/or a prescribed pressure term (\emph{P}). It has thus the following expression~\cite{Cia1988,PC1991}
\begin{equation}\label{eq:ext-power-surface}
  \mathcal{P}^{ext,s}(\ww) = \int_{\Sigma_{0}^{(DL)}} \left( \vec t_{0} \cdot \delta \varphi  \right) da_{0}- \int_{\Sigma^{(P)}} P \left(\nn \cdot \ww \right) da,
  \qquad
  (\delta \varphi =\ww \circ \varphi)
\end{equation}
where
\begin{itemize}
  \item the vector valued function $\vec t_{0}$, with values in $\espace$, is defined on the surface $\Sigma_{0}^{(DL)} \subset \partial \Omega_{0}$,
  \item the scalar function $P$ (the pressure) is defined on the surface $\Sigma^{(P)} \subset \partial \Omega_{\pp}$,
\end{itemize}
with the property that
\begin{equation*}
  \mathrm{int}\left(\Sigma^{(DL)}\right) \cap \mathrm{int}\left(\Sigma^{(P)}\right)= \emptyset, \quad   \text{and} \quad
  \Sigma^{(DL)} \cup \Sigma^{(P)} = \partial \Omega_{\pp} \setminus \partial \Omega_{\pp}^{\bxi},
\end{equation*}
where $\mathrm{int}(\Sigma)$ is the interior of the two-dimensional manifold (with boundary) $\Sigma$.

The virtual power of the pressure can be written as
\begin{equation*}
  \int_{\Sigma^{(P)}} P\, (\ww \cdot \nn)\, da = \int_{\Sigma^{(P)}} P\, i_{\ww}  \vol_{\bq},
\end{equation*}
using remark~\ref{rem:nds}. Now, by the change of variables formula, we get
\begin{equation}\label{eq:Pwn}
  \begin{aligned}
    \int_{\Sigma^{(P)}} P\, i_{\ww}  \vol_{\bq} = \int_{\pp^{-1}(\Sigma^{(P)})} \pp^{*}\left(P\, i_{\ww}  \vol_{\bq}\right)
     & = \int_{\pp^{-1}(\Sigma^{(P)})} \pp^{*}(P)\,  i_{\pp^{*}\ww} \, \pp^{*} \vol_{\bq}
    \\
     & = \int_{\Sigma_{\body}^{(P)}}(P\circ \pp)\,  i_{\bF^{-1}\delta \pp}\vol_{\bq}(  \bF \cdot ,  \bF \cdot, \bF \cdot)
    \\
     & = \int_{\Sigma_{\body}^{(P)}}(P\circ \pp)\,   \vol_{\bq} ( \delta \pp , \bF \cdot,\bF\cdot),
  \end{aligned}
\end{equation}
where $\bF=T\pp$, $\delta \pp=\ww \circ \pp$ and $\Sigma_{\body}^{(P)} := \pp^{-1}(\Sigma^{(P)})\subset \partial \body$.
Here, and in the following, given a $k$-form $\omega$, the notation $\omega(\bF \cdot ,  \dotsc, \bF \cdot)$
has to be understood as the $k$-form defined by
\begin{equation*}
  \omega(  \bF \cdot ,  \dotsc, \bF \cdot)(A_1,\dots ,A_k)=\omega(  \bF A_1 ,  \dotsc, \bF A_k).
\end{equation*}

\begin{rem}\label{rem:WP-on-the-body}
  If we identify the body $\body$ with a reference configuration $\Omega_{0}$, we get $\pp_{0}\equiv \id$, $\pp\equiv \varphi$,
  $\Sigma_{\body}^{(P)} \equiv\Sigma_{0}^{(P)}\subset \partial \Omega_{0}$, $\bF\equiv\bF_{\varphi}$ and $\delta \bxi \equiv \delta\varphi$, and
  \begin{align*}
    \int_{\Sigma_{\body}^{(P)}} (P\circ \pp)\, \vol_{\bq}(\delta \pp , \bF \cdot,\bF\cdot)
     & = \int_{\Sigma_{0}^{(P)}} (P\circ \varphi)\, \vol_{\bq}(\delta \bxi , \bF_\varphi \cdot,\bF_\varphi\cdot)
    \\
     & =
    \int_{\Sigma_{0}^{(P)}}  (P\circ \varphi)\, (\det \bF_{\varphi}) i_{{\bF_{\varphi}}^{-1}\delta\bxi} \vol_{\bq}
    \\
     & = \int_{\Sigma_{0}^{(P)}} (P\circ \varphi)\, J_{\varphi} {\bF_{\varphi}}^{-1}\delta\bxi \cdot \nn_{0} \, da_{0}.
  \end{align*}
  We recover the well-known expression of the pressure work
  \begin{equation*}
    \int_{\Sigma^{(P)}} P\, (\ww \cdot \nn)\, da = \int_{\Sigma_{0}^{(P)}} (P\circ \varphi)\, J_\varphi \bF_\varphi^{-1} \delta\bxi \cdot \nn_{0}\, da_{0},
    \qquad
    J_\varphi=\det \bF_\varphi,
  \end{equation*}
  formulated on the reference configuration $\Omega_{0}$ (see~\cite{Pea1956,Sew1965,Bal1976/77,Cia1988,MH1994}). This expression is usually obtained, using the Nanson formula~\eqref{eq:Nanson-formula}. One might note that
  \begin{equation*}
    J_{\varphi} {\bF_{\varphi}}^{-1}\delta\bxi \cdot \nn_{0} \, da_{0} = \vol_{\bq}(\delta \bxi , \bF_\varphi \cdot,\bF_\varphi\cdot).
  \end{equation*}
\end{rem}

Finally, introducing
\begin{equation*}
  \mathcal{W}_{\pp}^{ext,s}(\delta \pp) := -\mathcal{P}^{ext,s}(\delta \pp\circ\pp^{-1}),
\end{equation*}
we recast~\eqref{eq:ext-power-surface} on the body as
\begin{equation}\label{eq:ext-W-surface}
  \mathcal{W}_{\pp}^{ext,s}(\delta \pp) = -\int_{\Sigma_\body^{(DL)}} \left( \vec t_{0} \circ \pp_{0} \cdot \delta \pp \right) da_{\bgamma_{0}} + \int_{\Sigma_{\body}^{(P)}} (P\circ \pp)\,\vol_{\bq} ( \delta \pp , \bF \cdot,\bF\cdot),
\end{equation}
where we have set $\Sigma_{\body}^{(DL)}=\pp_{0}^{-1}(\Sigma_{0}^{(DL)})\subset \partial \body$ and $da_{\bgamma_{0}}=\pp_{0}^{*} \, da_{0}$. Thus, $\mathcal{W}^{ext,s}$ is naturally interpreted as a \emph{1-form} (see~\autoref{sec:differential-forms}) on the configuration space $\Emb(\body,\espace)$.

\begin{rem}[Equilibrium equations]
  The adequate choice of virtual displacement fields vanishing on the boundary $\partial \Omega_{\pp}$ classically leads to the equilibrium equations, but expressed here, either on the deformed configuration $\Omega_{\pp}$, as
  \begin{equation}\label{eq:equilOmega}
    \dive \bsigma + \ff_{v}=0,
  \end{equation}
  or, on the body $\body$, as
  \begin{equation}\label{eq:equilBody}
    \dive^{\bgamma} \msigma +\vec{\mathfrak{f}}_{v}=0,
    \qquad
    \vec{\mathfrak{f}}_{v} = \pp^{*}\ff_{v}=\bF^{-1} \ff_{v} \circ \pp,
  \end{equation}
  where $\vec{\mathfrak{f}}_{v}$ stands for the external forces density on $\body$ and where $\dive^{\bgamma} \msigma= \tr_{13} (\nabla^{\bgamma} \msigma) = \pp^* \dive \bsigma$, $\nabla^{\bgamma}$ being the Riemannian covariant derivative corresponding to the metric $\bgamma=\pp^* \bq$.
\end{rem}

\section{A general scheme for Lagrangian formulations of hyper-elasticity problems}
\label{sec:Lagrangian-formulations}

We now introduce a systematic method to build a Lagrangian (a potential energy) for an hyper-elasticity problem, when certain compatibility conditions are satisfied, concerning the surface forces. When these conditions are not satisfied, we formulate, anyway, \emph{non-holonomic} constraints, \textit{i.e.}, which involve not only the embedding $\pp$, but also its variation $\delta \pp$, to try to bypass these restrictions. This method is illustrated by applying them to two types of Neumann conditions: \emph{dead load (DL)} in section~\ref{subsec:dead-loads} and \emph{prescribed pressure (P)} in~\autoref{sec:prescribed-pressure}. These examples are classical, but the way we recast them seems to be original and illustrates the power of differential geometry (in infinite dimension) to tackle the problem in a more conceptual and systematic way (which avoid numerous typos found in the literature). Moreover,
the proposed method allows us to improve the non-holonomic constraints formulated by Beatty~\cite{Bea1970} to ensure the existence of a pressure potential. It is rather general and could, \textit{a priori}, be applied to many other situations. In a connected but different direction, the interested reader may refer to~\cite{PV1989,Via2013,Zdu2019} for discussions on the geometrization of point-wise holonomic constraints in finite strain theory.

The principle of virtual work, once recast on the body, involves various differential one-forms $\mathcal{W}^{k}$ defined on $\Emb(\body,\espace)$, like $\mathcal{W}^{ext,s}$ in~\eqref{eq:ext-W-surface} \cite{ES1980,Seg1986}. This principle stipulates that an embedding $\pp$ (which may be subject to some holonomic constraints) is solution of the mechanical problem if and only if the sum of these one-forms vanishes for all variations $\delta\pp$ (which may be subject to some non-holonomic constraints). The goal here is to try to recast this problem as a Lagrangian variational problem. In the best case, when each involved one-form $\mathcal{W}^{k}$ is exact, meaning that $\mathcal{W}^{k} = d \mathcal{L}^{k}$, then, each solution of the mechanical problem is an extremal of the Lagrangian
\begin{equation*}
  \mathcal{L} := \sum_{k} \mathcal{L}^{k}.
\end{equation*}

The good news is that classical differential geometry furnishes tools, like the Poincaré lemma (see~\autoref{sec:differential-forms}) to decide if the problem admits a Lagrangian, and in that case, to calculate a Lagrangian. Otherwise, it allows to formulate explicitly non-holonomic constraints under which a Lagrangian may still be defined. These tools can be extrapolated to differential geometry in infinite dimension and we will now describe them, in the mechanical situations we are interested in. The approach used here is the same as the one adopted by Arnold in~\cite{Arn1965}: use classical results from finite dimensional geometry, extrapolate them in this extended infinite dimensional setting, and then check that they are still true. This check being verified, in our case, by an explicit variation calculus.

The tool required to achieve our goal is the theory of \emph{differential forms} which is briefly recalled in~\autoref{sec:differential-forms} (see also~\cite{Car1970,AMR1988,Lan1999}). For a rigorous extension of this formalism to infinite dimensional spaces, one may consider useful to look at~\cite{Ham1982,Olv1993,KM1997}.
For our concern, one needs only to know that $\Emb(\body,\espace)$ can be considered as an open set of the vector space $\Cinf(\body,\espace)$, as soon as an origin has been chosen in $\espace$. Thus, a one-form on $\Emb(\body,\espace)$ is a continuous linear functional on $\Cinf(\body,\espace)$, depending smoothly on $\pp$, and a two-form on $\Emb(\body,\espace)$ is a skew-symmetric continuous bilinear functional on $\Cinf(\body, \espace)$, depending smoothly on $\pp$. Extrapolating the theory of differential forms to infinite dimension (see~\cite{Ham1982,KM1997}), we define:
\begin{itemize}
  \item the \emph{exterior derivative} of a $0$-form on $\Emb(\body,\espace)$ (\textit{i.e.}, a functional $\mathcal{L}$), just as the first variation of $\mathcal{L}$
        \begin{equation*}
          (d\mathcal{L})_{\pp}(\delta\pp) = \delta \mathcal{L} := \frac{d}{ds} \Big|_{s=0} \mathcal{L}(\pp(s)),
        \end{equation*}
        where $\pp(s)$ is a path of embeddings with $\pp(0)=\pp$ and $\dot \pp(0)=\delta \pp$;

  \item the \emph{exterior derivative} of a $1$-form $\mathcal{W}$ on $\Emb(\body,\espace)$ as
        \begin{equation*}
          (d\mathcal{W})_{\pp}(\delta_{1}\pp,\delta_{2}\pp) := \delta_{1}(\mathcal{W}_{p}(\delta_{2}\pp)) - \delta_{2}(\mathcal{W}_{p}(\delta_{1}\pp)),
        \end{equation*}
        where
        \begin{equation*}
          \delta_{2}(\mathcal{W}_{p}(\delta_{1}\pp)) := \frac{d}{ds} \Big|_{s=0} \mathcal{W}_{p(s)}(\delta_{1}\pp),
        \end{equation*}
        and $\pp(s)$ is a path of embeddings with $\pp(0)=\pp$ and $\dot \pp(0)=\delta_{2} \pp$.
\end{itemize}

\begin{rem}\label{rem-vec-val-p}
  To write these two formulas, we have used the hypothesis that $\espace$ is an affine space, and thus that $T\Emb(\body,\espace)=\Emb(\body,\espace) \times \Cinf(\body,E)$ is trivial.
\end{rem}

\subsection{The Poincaré integrator}
\label{subsec:Poincare-integrator}

Our methodology is the following. Starting from a single $1$-form $\mathcal{W}$ in the list $(\mathcal{W}^{k})$, we check first if it is is closed. In other words, we check if
\begin{equation*}
  (d\mathcal{W})_{\pp}(\delta_{1}\pp,\delta_{2}\pp) = \delta_{1}(\mathcal{W}_{p}(\delta_{2}\pp)) - \delta_{2}(\mathcal{W}_{p}(\delta_{1}\pp)) = 0,
\end{equation*}
with then two cases.
\begin{itemize}
  \item \textbf{The case $\mathcal{W}$ is closed} ($d\mathcal{W}=0$). Then, a Lagrangian $\mathcal{L}$ is obtained locally, using the Poincaré lemma~\ref{lem:Poincare}, by the following procedure. Let $\pp_{0} \in \Emb(\body,\espace)$ be a reference configuration. The displacement
        \begin{equation*}
          \bxi(\pp) = \pp - \pp_{0},
        \end{equation*}
        corresponds to the radial vector field issued from $\pp_{0}$ on $\Emb(\body,\espace)$. Consider now the neighborhood of $\pp_{0}$, defined by
        \begin{equation}\label{eq:Up0}
          \mathcal{U}_{\pp_{0}} := \set{\pp \in \Cinf(\body,\espace); \; \norm{\bF - \bF_{0}}_{0} < 1},
        \end{equation}
        where
        \begin{equation*}
          \norm{\bF - \bF_{0}}_{0} := \sup_{\XX\in \body} \norm{\bF(\XX) - \bF_{0}(\XX)}.
        \end{equation*}
        The set $\mathcal{U}_{\pp_{0}}$ is an \emph{open convex neighborhood} of $\pp_{0}$ in $\Cinf(\body,\espace)$, which is contained in $\Emb(\body,\espace)$ (see~\autoref{sec:Frechet-topology}). Then, the Poincaré formula~\eqref{eq:Poincare-formula} provides us with the primitive
        \begin{equation*}
          \mathcal{L}(\pp) = \int_{-\infty}^{0} \left[(\phi^{t})^{*} i_{\bxi}\mathcal{W}\right](\pp) \, dt ,
        \end{equation*}
        defined on $\mathcal{U}_{\pp_{0}}$, and where $\phi^{t}(\pp) = e^{t} \pp+(1-e^{t}) \pp_{0}$ is the flow of the radial field $\bxi(\pp)=\pp-\pp_{0}$. But
        \begin{equation*}
          \left[(\phi^{t})^{*} i_{\bxi}\mathcal{W}\right](\pp) = \mathcal{W}_{\phi^{t}(\pp)}(\bxi(\phi^{t}(\pp))) = e^{t}\mathcal{W}_{\phi^{t}(\pp)}(\bxi(\pp)),
        \end{equation*}
        because $\bxi(\phi^{t}(\pp)) = e^{t} \bxi(\pp)$ and we get finally the sought Lagrangian for $\mathcal{W}$
        \begin{equation}\label{eq:Poincare-primitive}
          \mathcal{L}(\pp) = \int_{-\infty}^{0} e^{t}\mathcal{W}_{\phi^{t}(\pp)}(\bxi(\pp)) \, dt.
        \end{equation}

  \item \textbf{The case $\mathcal{W}$ is not closed} ($d\mathcal{W}\neq 0$). Then, by~\eqref{eq:Poincare-formula-full}, we get
        \begin{equation*}
          d\mathcal{L} = \mathcal{W} - \int_{-\infty}^{0} (\phi^{t})^{*} i_{\bxi}d\mathcal{W} \, dt,
        \end{equation*}
        where
        \begin{align*}
          \left[(\phi^{t})^{*} i_{\bxi}d\mathcal{W}\right]_{\pp}(\delta\pp) & = (d\mathcal{W})_{\phi^{t}(\pp)}(T\phi^{t}.\, \bxi(\pp), T\phi^{t}.\delta\pp)
          \\
                                                                            & = e^{2t}(d\mathcal{W})_{\phi^{t}(\pp)}(\bxi(\pp), \delta\pp),
        \end{align*}
        because the linear tangent mapping $T_{\pp}\phi^{t}: T_{\pp}\Emb(\body,\espace) \to T_{\phi^{t}(\pp)}\Emb(\body,\espace)$ can be written as
        \begin{equation*}
          T_{\pp}\phi^{t}.\delta\pp = e^{t} \delta\pp.
        \end{equation*}
        In this case, the integral
        \begin{equation}\label{eq:Poincare-obstruction}
          I = \int_{-\infty}^{0} e^{2t}(d\mathcal{W})_{\phi^{t}(\pp)}(\bxi(\pp), \delta\pp) \, dt,
        \end{equation}
        when non vanishing, is the obstruction for~\eqref{eq:Poincare-primitive} to be a primitive of $\mathcal{W}$ and $I = 0$ can be interpreted as a \emph{non-holonomic constraint}, required for~\eqref{eq:Poincare-primitive} to become a Lagrangian for $\mathcal{W}$.
\end{itemize}

\subsection{Example: Dead loads}
\label{subsec:dead-loads}

Before addressing the more involving case of a prescribed pressure, let us illustrate our methodology on the well-known case of a dead load.

\begin{rem}
  The related questions of how to recast locally such a boundary condition on the body $\body$, in a so-called intrinsic manner~\cite{Nol1972,Rou1980,GLM2014,Fia2016}, and of how to redefine a first Piola--Kirchhoff tensor on $\body$ are answered in \autoref{sec:PK1}.
\end{rem}

In the dead load case (see~\eqref{eq:ext-W-surface}), the corresponding one-form defined on $\Emb(\body,\espace)$ can be written as
\begin{equation*}
  \mathcal{W}^{DL}_{\pp}(\delta\pp) = - \int_{\Sigma_{\body}^{(DL)}} (\delta\pp \cdot \vec t_{0}\circ p_{0} )\,da_{\bgamma_{0}}.
\end{equation*}
This form is obviously closed, since it does not depend explicitly on $\pp$. Thus, a Lagrangian (\textit{i.e.}, a potential for the applied surface force~\cite{Cia1988}) $\mathcal{L}^{DL}(\pp)$ exists for dead loads and is given by~\eqref{eq:Poincare-primitive}. After integration, it can be written as
\begin{equation*}
  \mathcal{L}^{DL}(\pp) = - \int_{\Sigma_{\body}^{(DL)}} ((\pp-\pp_{0})\cdot \vec t_{0}\circ p_{0} )\,da_{\bgamma_{0}},
\end{equation*}
since
\begin{equation*}
  \mathcal{W}^{DL}_{\phi^{t}(\pp)}(\bxi(\pp)) = - \int_{\Sigma_{\body}^{(DL)}} ((\pp - \pp_{0}) \cdot \vec t_{0}\circ p_{0}) \,da_{\bgamma_{0}}.
\end{equation*}

\section{Prescribed pressure}
\label{sec:prescribed-pressure}

The question of the existence of a Lagrangian (a potential energy) for hyper-elasticity problems with boundary conditions of the prescribed pressure type
\begin{equation}\label{eq:CLP}
  \left.\widehat{\bsigma}\nn\right|_{\Sigma} = - P \nn,
\end{equation}
on some part $\Sigma=\Sigma^{(P)}$ of the boundary $\partial\Omega_{\pp}$ has been addressed in several works~\cite{Pea1956,Sew1965,Sew1967,Bea1970,Bal1976/77} but rarely (never ?) directly on the body's boundary $\partial \body$. The prescribed pressure boundary condition~\eqref{eq:CLP} becomes, by pullback on $\partial \body$,
\begin{equation}\label{eq:CLP-body}
  \left.\widehat{\msigma}\bN \right|_{\Sigma_\body} = - (P\circ \pp) \bN,
  \qquad
  \Sigma_\body = \pp^{-1}(\Sigma),
\end{equation}
where $\widehat{\msigma} = \pp^{*}\widehat{\bsigma}$ is the Noll stress tensor in its mixed form, and where $\pp^{*} \nn=\bN$ is the unit normal on $\partial \body$ for the metric $\bgamma=\pp^{*}\bq$. In this section, we assume that $P$ is uniform, \textit{i.e.}, $P = P_{k}$, on each component $\Sigma^{k}$ of the boundary $\Sigma^{(P)}\subset \partial \Omega_{\pp}$, where the pressure is applied. The contribution to the virtual power of exterior forces $\mathcal{P}^{ext}$ of these prescribed pressures can thus be written as
\begin{equation}\label{eq:CondLP}
  -\sum_{k} P_{k} \int_{\Sigma^{k}} (\ww \cdot \nn)\, da,
\end{equation}
where $\ww$ is the virtual Eulerian velocity.

\subsection{The pressure virtual power is not closed}
\label{subsec:pressure-not-closed}

The integral
\begin{equation*}
  \int_{\Sigma} P(\ww \cdot \nn) da
\end{equation*}
was recast on the body boundary in~\eqref{eq:Pwn}. Here, for simplicity, we will assume that $P=1$ is a constant function and we have thus
\begin{equation}\label{eq:pressure-power}
  \mathcal{W}^{P}_{\pp}(\delta \pp):= \int_{\Sigma} (\ww \cdot \nn)\, da = \int_{\Sigma_{\body}} \vol_{\bq} ( \delta \pp , \bF \cdot,\bF\cdot)
\end{equation}

The pressure case illustrates the second case, $\mathcal{W}$ is not closed, of our general scheme in~\autoref{subsec:Poincare-integrator}.

\begin{lem}\label{lem:W-not-closed}
  The differential form
  \begin{equation*}
    \mathcal{W}^{P}_{\pp}(\delta \pp) = \int_{\Sigma_{\body}} \vol_{\bq} ( \delta \pp , \bF \cdot,\bF\cdot)
  \end{equation*}
  defined on $\Emb(\body,\espace)$ is not closed. Its exterior derivative has for expression
  \begin{equation}\label{eq:dWbody}
    (d\mathcal{W}^{P})_{\pp}(\delta \pp_{1}, \delta \pp_{2}) = \int_{\partial \Sigma_{\body}} \left( \delta \pp_{2} \times \delta \pp_{1} \right) \cdot \, \bF\, d\vec \ell_{\body},
  \end{equation}
  where $d\vec \ell_{\body}=\pp_{0}^{*}\,d\vec \ell_{0}$ is the oriented length element on $\partial \Sigma_{\body}$.
\end{lem}

The proof we propose here is based on the observation that an embedding $\pp : \body \to \espace$ and a variation of this embedding $\delta \pp: \body \to T\espace=\espace \times E$ can be considered as \emph{vector valued functions} $\ff: \body \to E$ (see \autoref{sec:kinematics}). Now, if $A$ is a vector field on $\body$, we can define the Lie derivative $\Lie_{A} \ff$ of $\ff$ relative to $A$. But, since $\ff$ is considered as a function, it depends only on the pointwise value of $A$ (and not of its first derivatives) and we get moreover
\begin{equation*}
  T\ff.A = \Lie_{A}\ff.
\end{equation*}

\begin{proof}[Proof of lemma~\ref{lem:W-not-closed}]
  We have
  \begin{equation*}
    \vol_{\bq}(\delta_{1}\pp, \bF A , \bF B) = \vol_{\bq}(\delta_{1}\pp, \Lie_{A}\pp , \Lie_{B}\pp),
  \end{equation*}
  $A$ and $B$ being vector fields on $\partial\body$ (which do not depend on $\pp$). Let $\delta_{2}\pp$ be a second variation of $\pp$, we get thus
  \begin{equation*}
    \delta_{2} (\vol_{\bq}(\delta_{1}\pp, \Lie_{A}\pp , \Lie_{B}\pp)) = \vol_{\bq}(\delta_{1}\pp, \Lie_{A}\delta_{2}\pp , \Lie_{B}\pp) + \vol_{\bq}(\delta_{1}\pp, \Lie_{A}\pp , \Lie_{B}\delta_{2}\pp),
  \end{equation*}
  and hence
  \begin{align*}
    (d\mathcal{W}^{P})_{\pp}(\delta \pp_{1}, \delta \pp_{2}) & = \int_{\Sigma_{\body}} \delta_{1}\vol_{\bq}(\delta_{2} \pp, \bF \cdot , \bF \cdot) -\delta_{2}\vol_{\bq}(\delta_{1} \pp, \bF \cdot , \bF \cdot),
  \end{align*}
  with
  \begin{align*}
    \delta_{1}\vol_{\bq}(\delta_{2} \pp, \bF A , \bF B) -\delta_{2}\vol_{\bq}(\delta_{1} \pp, \bF A , \bF B)
     & = \vol_{\bq}(\delta_{2}\pp, \Lie_{A}\delta_{1}\pp , \Lie_{B}\pp) + \vol_{\bq}(\delta_{2}\pp, \Lie_{A}\pp , \Lie_{B}\delta_{1}\pp)
    \\
     & \quad -\vol_{\bq}(\delta_{1}\pp, \Lie_{A}\delta_{2}\pp , \Lie_{B}\pp) - \vol_{\bq}(\delta_{1}\pp, \Lie_{A}\pp , \Lie_{B}\delta_{2}\pp) .
  \end{align*}

  Let us now introduce the $1$-form $\alpha$ on $\partial\body$ defined by
  \begin{equation*}
    \alpha(A) := \vol_{\bq}(\delta_{2}\pp, \delta_{1}\pp, \Lie_{A}\pp).
  \end{equation*}
  By lemma~\ref{lem:alpha-exterior-derivative}, we have then
  \begin{align*}
    d\alpha(A,B) & = \vol_{\bq}(\Lie_{A}\delta_{2}\pp, \delta_{1}\pp, \Lie_{B}\pp) - \vol_{\bq}(\Lie_{B}\delta_{2}\pp, \delta_{1}\pp, \Lie_{A}\pp)
    \\
                 & \quad + \vol_{\bq}(\delta_{2}\pp, \Lie_{A}\delta_{1}\pp, \Lie_{B}\pp) - \vol_{\bq}(\delta_{2}\pp, \Lie_{B}\delta_{1}\pp. \Lie_{A}\pp),
  \end{align*}
  Hence, using Stokes theorem, we get
  \begin{equation*}
    (d\mathcal{W}^{P})_{\pp}(\delta \pp_{1}, \delta \pp_{2}) = \int_{\Sigma_{\body}} d \alpha = \int_{\partial\Sigma_{\body}} \alpha,
  \end{equation*}
  and thus
  \begin{equation*}
    (d\mathcal{W}^{P})_{\pp}(\delta \pp_{1}, \delta \pp_{2}) = \int_{\partial\Sigma_{\body}}  \vol_{\bq}(\delta_{2}\pp, \delta_{1}\pp, \bF \cdot)
    =  \int_{\partial\Sigma_{\body}}  (\delta_{2}\pp \times \delta_{1}\pp) \cdot \bF d\vec \ell_{\body}.
  \end{equation*}
\end{proof}

\subsection{A pressure potential under non-holonomic constraints}
\label{subsec:pressure-potential}

Since $\mathcal{W}^{P}$ is not closed, a Lagrangian for the prescribed pressure virtual power does not exist \textit{a priori} \cite{Pea1956,Sew1965,Sew1967,Bea1970,Bal1976/77,PC1991}. However, if we define the functional $\mathcal{L}^{P}$ by~\eqref{eq:Poincare-primitive}, then we have
\begin{equation*}
  d\mathcal{L}^{P}(\delta \pp) = \mathcal{W}^{P}_{p}(\delta \pp),
\end{equation*}
\emph{provided} that condition~\eqref{eq:Poincare-obstruction} is satisfied. We will now achieve the calculation of $\mathcal{L}^{P}$ and the corresponding non-holonomic constraints. First, we have
\begin{equation*}
  \mathcal{W}^{P}_{\phi^{t}(\pp)}(\bxi(\pp)) = \int_{\Sigma_{\body}} \vol_{\bq} ( \delta \pp , T\phi^{t} \cdot,T\phi^{t} \cdot)
\end{equation*}
and~\eqref{eq:Poincare-primitive} can be written as
\begin{equation*}
  \mathcal{L}^{P}(p) = \int_{-\infty}^{0} e^{t} \left(\int_{\Sigma_{\body}} \vol_{\bq}\big( \bxi(\pp), (e^{t} \bF+(1-e^{t}) \bF_{0}) \cdot, (e^{t} \bF+(1-e^{t}) \bF_{0}) \cdot\big)\right) \, dt ,
\end{equation*}
since $T\phi^{t}(\pp) = e^{t} \bF+(1-e^{t}) \bF_{0}$. The calculation is then straightforward, making use of the multi-linearity of $\vol_{\bq}$, and leads to the Lagrangian (the pressure potential)
\begin{equation*}
  \mathcal{L}^{P}(\pp) = \frac{1}{6} \int_{\Sigma_{\body}} 2\vol_{\bq}(\bxi, \bF \cdot , \bF \cdot)
  + \left(\vol_{\bq}(\bxi, \bF \cdot , \bF_{0} \cdot)+\vol_{\bq}(\bxi, \bF_{0} \cdot , \bF \cdot)\right)
  + 2\vol_{\bq}(\bxi, \bF_{0} \cdot , \bF_{0} \cdot).
\end{equation*}

The non-holonomic constraints are obtained by specifying~\eqref{eq:Poincare-obstruction} for $\mathcal{W}=\mathcal{W}^{P}$ and thus by computing the integral
\begin{equation*}
  \int_{-\infty}^{0} e^{2t}(d\mathcal{W}^{P})_{\phi^{t}(\pp)}(\bxi(\pp), \delta\pp) \, dt
\end{equation*}
where
\begin{equation*}
  (d\mathcal{W}^{P})_{\pp}(\delta \pp_{1}, \delta \pp_{2}) = \int_{\partial \Sigma_{\body}} \left( \delta \pp_{2} \times \delta \pp_{1} \right) \cdot \, \bF\, d\vec \ell_{\body}.
\end{equation*}
We get
\begin{equation*}
  \oint_{\partial \Sigma_{\body}} (\delta \pp \times \bxi(\pp))\cdot
  \left(  \int_{-\infty}^{0} ( e^{3t} \bF+e^{2t} (1-e^{t}) \bF_{0})\, dt\right) d\vec \ell_{\body},
\end{equation*}
which is equal, after integration to
\begin{equation*}
  - \frac{1}{6}\oint_{\partial \Sigma_{\body}} (\bxi(\pp) \times \delta \pp)\cdot (2\bF+\bF_{0}) d\vec \ell_{\body}.
\end{equation*}
Thus, the non-holonomic constraints, which must be satisfied, in order for $\mathcal{L}^{P}$ to be a Lagrangian for $\mathcal{W}^P$ write
\begin{equation*}
  \oint_{\partial \Sigma_{\body}} (\bxi \times \delta \bxi)\cdot (2\bF+\bF_{0}) d\vec \ell_{\body} = 0.
\end{equation*}

We will summarize these results in the following theorem, checking, this time, by a direct variation calculus,
that $\mathcal{L}^{P}$ is a \emph{potential} for $\mathcal{W}^P$ which is globally defined.
Indeed, it states that the functional \eqref{eq:KD-pressure-potential} is a Lagrangian for the pressure boundary term variational problem, provided that the non-holonomic condition \eqref{eq:non-holonomic-condition-body} is satisfied. Remark that
so far, this was established only locally, on the neighborhood
$\mathcal{U}_{\pp_{0}}$ of the reference configuration $\pp_{0}$ defined by~\eqref{eq:Up0}.

\begin{thm}\label{thm:delta-LP}
  Let us consider the functional
  \begin{equation}\label{eq:KD-pressure-potential}
    \mathcal{L}^{P}(\pp) = \frac{1}{6} \int_{\Sigma_{\body}} 2\,\vol_{\bq}(\bxi, \bF \cdot , \bF \cdot)
    + \left(\vol_{\bq}(\bxi, \bF \cdot , \bF_{0} \cdot) + \vol_{\bq}(\bxi, \bF_{0} \cdot , \bF \cdot)\right)
    + 2\,\vol_{\bq}(\bxi, \bF_{0} \cdot , \bF_{0} \cdot),
  \end{equation}
  defined on $\Emb(\body,\espace)$, where $\bxi(\pp) := \pp-\pp_{0}$ is the displacement field. Then
  \begin{equation}\label{eq:derivative-L-KD}
    d\mathcal{L}^{P}(\delta\pp) = \mathcal{W}^{P}(\delta\pp) + \frac{1}{6}\oint_{\partial \Sigma_{\body}} (\bxi \times \delta\bxi)\cdot (2\bF+\bF_{0}) d\vec \ell_{\body},
  \end{equation}
  $d\vec \ell_{\body}$ being the vector length element on $\partial \Sigma_\body$. In particular, the condition for the functional $\mathcal{L}^{P}$ to be a pressure potential (\emph{i.e.}, a primitive of $\mathcal{W}^{P}$) is thus
  \begin{equation}\label{eq:non-holonomic-condition-body}
    \oint_{\partial \Sigma_{\body}} (\bxi \times \delta\bxi)\cdot (2\bF+\bF_{0}) d\vec \ell_{\body} = 0.
  \end{equation}
\end{thm}

\begin{proof}[Proof of theorem \ref{thm:delta-LP}]
  Set
  \begin{equation}\label{eq:omegai}
    \begin{aligned}
      \omega^{1} & := 2\,\vol_{\bq}(\bxi, \bF \cdot , \bF \cdot)
      \\
      \omega^{2} & := \vol_{\bq}(\bxi, \bF \cdot , \bF_{0} \cdot) + \vol_{\bq}(\bxi, \bF_{0} \cdot ,\bF \cdot)
      \\
      \omega^{3} & := 2\,\vol_{\bq}(\bxi, \bF_{0} \cdot , \bF_{0} \cdot),
    \end{aligned}
  \end{equation}
  and $\omega := \omega^{1} + \omega^{2} + \omega^{3}$, so that
  \begin{equation}\label{eq:dLp}
    \mathcal{L}^{P} = \frac{1}{6} \int_{\Sigma_{\body}} \omega , \qquad \delta\mathcal{L}^{P} = \frac{1}{6} \int_{\Sigma_{\body}} \delta\omega .
  \end{equation}
  We have
  \begin{align*}
    \delta\omega^{1}(A,B) & = 2\,\left(\vol_{\bq}(\delta\pp, \Lie_{A}\pp , \Lie_{B}\pp) + \vol_{\bq}(\bxi, \Lie_{A}\delta\pp , \Lie_{B}\pp) + \vol_{\bq}(\bxi, \Lie_{A}\pp , \Lie_{B}\delta\pp)\right),
    \\
    \delta\omega^{2}(A,B) & = \vol_{\bq}(\delta\pp, \Lie_{A}\pp , \Lie_{B}\pp_{0}) + \vol_{\bq}(\bxi, \Lie_{A}\delta\pp , \Lie_{B}\pp_{0})
    \\
                          & \quad + \vol_{\bq}(\delta\pp, \Lie_{A}\pp_{0} , \Lie_{B}\pp) + \vol_{\bq}(\bxi, \Lie_{A}\pp_{0} , \Lie_{B}\delta\pp),
    \\
    \delta\omega^{3}(A,B) & = 2\,\vol_{\bq}(\delta\pp, \Lie_{A}\pp_{0} , \Lie_{B}\pp_{0}),
  \end{align*}
  and thus
  \begin{align*}
    \delta\omega (A,B) & = \vol_{\bq}(\delta\pp, \Lie_{A}\pp , \Lie_{B}(2\pp +\pp_{0})) + \vol_{\bq}(\delta\pp, \Lie_{A}\pp_{0} , \Lie_{B}(2\pp_{0} + \pp))
    \\
                       & \quad + \vol_{\bq}(\bxi, \Lie_{A}\delta\pp , \Lie_{B}(2\pp +\pp_{0})) - \vol_{\bq}(\bxi, \Lie_{B}\delta\pp, \Lie_{A}(2\pp + \pp_{0})).
  \end{align*}
  Now we need to do an integration by part to get rid of the terms that contain Lie derivatives of $\delta\pp$. To do this, we introduce the $1$-form $\alpha$ on $\partial\body$ defined by
  \begin{equation*}
    \alpha(A) := \vol_{\bq}(\bxi, \delta\pp, \Lie_{A}(2\pp + \pp_{0})).
  \end{equation*}
  such that, by lemma \ref{lem:alpha-exterior-derivative},
  \begin{align*}
    d\alpha(A,B) & = \vol_{\bq}(\Lie_{A}\bxi, \delta\pp, \Lie_{B}(2\pp + \pp_{0})) - \vol_{\bq}(\Lie_{B}\bxi, \delta\pp, \Lie_{A}(2\pp + \pp_{0}))
    \\
                 & \quad + \vol_{\bq}(\bxi, \Lie_{A}\delta\pp, \Lie_{B}(2\pp + \pp_{0})) - \vol_{\bq}(\bxi, \Lie_{B}\delta\pp, \Lie_{A}(2\pp + \pp_{0})).
  \end{align*}
  Therefore, after simplification, we have
  \begin{align*}
    \delta\omega (A,B) & = 6 \vol_{\bq}(\delta\pp, \Lie_{A}\pp, \Lie_{B}\pp) + d\alpha(A,B),
  \end{align*}
  and thus
  \begin{equation*}
    \delta \mathcal{L}^{P} = \frac{1}{6} \int_{\Sigma_{\body}} \delta\omega = \int_{\Sigma_{\body}} \vol_{\bq}(\delta\pp, \bF \cdot , \bF\cdot )+ \frac{1}{6}  \int_{\partial \Sigma_{\body}} \alpha,
  \end{equation*}
  where the restriction of $\alpha$ to $\partial \Sigma_{\body}$ can be written as
  \begin{equation*}
    \alpha = \vol_{\bq}(\bxi, \delta\pp, (\bF_{0}+ 2 \bF) \cdot )  = (\bxi \times \delta\bxi)\cdot (\bF_{0}+ 2 \bF) d \vec \ell_{\body},
  \end{equation*}
  $d\vec \ell_{\body}$ being the vector length element relative to the metric $\bgamma_{0} = \pp_{0}^{*}\bq$ on $\body$. Finally
  \begin{equation*}
    \delta \mathcal{L}^{P} = \mathcal{W}^{P}(\delta\pp) + \frac{1}{6}\oint_{\partial \Sigma_{\body}} (\bxi \times \delta\pp)\cdot (2\bF+\bF_{0}) d\vec \ell_{\body},
  \end{equation*}
  which ends the proof.
\end{proof}

\section{Link with other existing formulations}
\label{sec:other-formulations}

The goal of this section is to relate the present work with other existing studies concerning the formulation of a pressure potential and the required constraints. These results are all expressed on a reference configuration $\Omega_{0}$ and not on the body $\body$.

\subsection{Recovering Pearson--Sewell potential and Beatty conditions}
\label{subsec:Sewell-Beatty}

The first formulation of a pressure potential seems to have been produced in 1956 by Pearson~\cite[Eq. (25) on p. 142]{Pea1956} and then reobtained by Sewell ten years later (see~\cite[Eq. (32) on p. 407]{Sew1965} and~\cite[Eq. (89) on p. 341]{Sew1967}). In all these works, the potential is written in components. Its intrinsic expression seems to have been given for the first time by Beatty in~\cite[Eq. (4.3) on p. 373]{Bea1970} but with some typos. Its (corrected) expression is:
\begin{equation}\label{eq:LP-Beatty}
  \mathcal{L}^{P}(\varphi) = \frac{P}{3} \int_{\Sigma_{0}} \left(J_{\varphi} {\bF_{\varphi}}^{-1}\, \bxi
  + \frac{1}{2}\Big((\tr \bF_{\varphi})\bxi - \bF_{\varphi} \bxi \Big) + \bxi \right) \cdot \nn_{0} \, da_{0},
\end{equation}
where $\bxi = \varphi-\id$. This expression corresponds to the potential $\mathcal{L}^{P}$ given by~\eqref{eq:KD-pressure-potential}, if we identify the body $\body$ with a reference configuration $\Omega_{0}$, embedded in Euclidean space $\espace$  (see below for a proof). Indeed, then, we make the identifications:
\begin{equation}\label{eq:Identif-B-Ref}
  \Sigma_{\body} = \Sigma_{0}, \qquad \pp\equiv \varphi, \qquad \bF_{0}\equiv \Id, \qquad \bF\equiv \bF_{\varphi}, \qquad \bxi \equiv \varphi - \id.
\end{equation}

Then, our non-holonomic constraints~\eqref{eq:non-holonomic-condition-body}, formulated on the body, recast on $\Omega_{0}$ as
\begin{equation}\label{eq:non-holonomic-condition-ref}
  \oint_{\partial \Sigma_{0}}  \left(\bxi \times \delta \bxi \right)
  \cdot \,( 2 \bF_\varphi +\Id) d\vec \ell_{0}=0,
  \qquad
  \bxi= \varphi-\id.
\end{equation}
This is an improvement compared to Beatty conditions~\cite[Eq. (4.6) on p. 374]{Bea1970},
\begin{equation}\label{eq:Beatty-conditions}
  \oint_{\partial \Sigma_{0}}  \left(\bxi \times \delta \bxi \right)
  \cdot \,d\vec \ell_{0}=0
  \qquad \textrm{and} \qquad \oint_{\partial \Sigma_{0}}   \left(\bxi \times \delta \bxi \right)
  \cdot \,\bF_{\varphi} d\vec \ell_{0}=0,
\end{equation}
which are stronger since~\eqref{eq:Beatty-conditions} implies \eqref{eq:non-holonomic-condition-ref}, but the converse does not hold.

\begin{rem}
  Both conditions~\eqref{eq:non-holonomic-condition-ref} and \eqref{eq:Beatty-conditions} are satisfied, in particular, when the variations $\delta \bxi$ vanish on the closed contour $\partial \Sigma_{0}$. More generally, in order for them to be verified it is sufficient that the virtual displacement $\delta \pp=\delta \bxi \in T\Emb(\body, \espace)$ remains collinear to the displacement $\bxi=\pp-\pp_{0}$ (\emph{i.e.}, $\bxi \times \delta\bxi=0$) all along the contour  $\partial \Sigma_{\body}$, a mechanistic condition indeed. Note finally that the constraints \eqref{eq:non-holonomic-condition-ref} and \eqref{eq:non-holonomic-condition-body} are trivially satisfied when $\partial \Sigma=\pp(\partial \Sigma_{\body}) =\emptyset$, as in the case of an uniform pressure applied on the entire external surface of a structure or on the entire surface of a fully embedded cavity.
\end{rem}

We conclude by providing a detailed calculation of how~\eqref{eq:KD-pressure-potential} recasts as~\eqref{eq:LP-Beatty} when we identify the body $\body$ with $\Omega_{0}$ embedded in Euclidean space $\espace$. Thanks to the identifications \eqref{eq:Identif-B-Ref}, $\mathcal{L}^P$ can be written as
\begin{equation*}
  \mathcal{L}^P(\varphi) = \frac{P}{6}\int_{\Sigma_{0}} \omega^{1} + \omega^{2} + \omega^{3},
\end{equation*}
where, according to the definition of the 2-forms $\omega^{i}$ by \eqref{eq:omegai},
\begin{align*}
  \omega^{1} & = 2\vol_{\bq}(\bxi, \bF_{\varphi} \cdot , \bF_{\varphi} \cdot) = 2(\det \bF_{\varphi}) i_{{\bF_{\varphi}}^{-1}\bxi} \vol_{\bq} = 2J_{\varphi} ({\bF_{\varphi}}^{-1}\bxi \cdot \nn_{0}) \, da_{0},
  \\
  \omega^{2} & = \vol_{\bq}(\bxi, \bF_{\varphi} \cdot , \cdot) + \vol_{\bq}(\bxi, \cdot , \bF_{\varphi} \cdot) = \left((\tr \bF_{\varphi})\bxi - \bF_{\varphi} \bxi\right) \cdot \nn_{0} \, da_{0},
  \\
  \omega^{3} & = 2\vol_{\bq}( \bxi, \cdot , \cdot) = i_{\bxi} \vol_{\bq} = 2(\bxi \cdot \nn_{0}) \, da_{0}.
\end{align*}
The first and third equalities are straightforward using~\eqref{eq:nds}. The second equality results from the following observation. The $2$-form $\omega^{2}$ is proportional to the area element $da_{0} = i_{\nn_{0}} \vol_{\bq}$ of $\partial\Omega_{0}$ (see \autoref{sec:differential-forms}). Therefore, if $\ee_{1},\ee_{2}$ is a direct orthonormal basis of $T_{\xx_{0}}\partial\Omega_{0}$, the tangent space to the boundary of $\Omega_{0}$, we have $\omega^{2} = \lambda \, da_{0}$, where $\lambda = \omega^{2}(\ee_{1},\ee_{2})$. Thus, writing $\bxi = \xi^{1}\ee_{1} + \xi^{2}\ee_{2} + (\bxi\cdot\nn_{0})\nn_{0}$, we get
\begin{align*}
  \omega^{2}(\ee_{1},\ee_{2}) & = \vol_{\bq}(\bxi, \bF_{\varphi} \ee_{1} , \ee_{2}) + \vol_{\bq}(\bxi, \ee_{1} , \bF_{\varphi} \ee_{2})
  \\
                              & = (\bxi\cdot\nn_{0})\vol_{\bq}(\nn_{0}, \bF_{\varphi} \ee_{1} , \ee_{2}) + (\bxi\cdot\nn_{0}) \vol_{\bq}(\nn_{0}, \ee_{1} , \bF_{\varphi} \ee_{2})
  \\
                              & \quad + \xi^{1} \vol_{\bq}(\ee_{1}, \bF_{\varphi} \ee_{1} , \ee_{2}) + \xi^{2}\vol_{\bq}(\ee_{2}, \ee_{1} , \bF_{\varphi} \ee_{2})
  \\
                              & = (\bxi\cdot\nn_{0})\left\{\vol_{\bq}(\nn_{0}, \bF_{\varphi} \ee_{1} , \ee_{2}) + \vol_{\bq}(\nn_{0}, \ee_{1} , \bF_{\varphi} \ee_{2}) + \vol_{\bq}(\bF_{\varphi} \nn_{0}, \ee_{1} , \ee_{2}) \right\}
  \\
                              & \quad + \xi^{1} \vol_{\bq}(\ee_{1}, \bF_{\varphi} \ee_{1} , \ee_{2}) + \xi^{2}\vol_{\bq}(\ee_{2}, \ee_{1} , \bF_{\varphi} \ee_{2}) - (\bxi\cdot\nn_{0})\vol_{\bq}(\bF_{\varphi} \nn_{0}, \ee_{1} , \ee_{2}).
\end{align*}
But
\begin{equation*}
  \vol_{\bq}(\nn_{0}, \bF_{\varphi} \ee_{1} , \ee_{2}) + \vol_{\bq}(\nn_{0}, \ee_{1} , \bF_{\varphi} \ee_{2}) + \vol_{\bq}(\bF_{\varphi} \nn_{0}, \ee_{1} , \ee_{2}) = \tr \bF_{\varphi}
\end{equation*}
and
\begin{multline*}
  \xi^{1} \vol_{\bq}(\ee_{1}, \bF_{\varphi} \ee_{1} , \ee_{2}) + \xi^{2}\vol_{\bq}(\ee_{2}, \ee_{1} , \bF_{\varphi} \ee_{2}) - (\bxi\cdot\nn_{0})\vol_{\bq}(\bF_{\varphi} \nn_{0}, \ee_{1} , \ee_{2})
  \\
  = - \vol_{\bq}(\bF_{\varphi}\bxi, \ee_{1} , \ee_{2}) = - \bF_{\varphi}\bxi\cdot\nn_{0},
\end{multline*}
because $\vol_{\bq}(\nn_{0}, \ee_{1} , \ee_{2})=1$. Hence
\begin{equation*}
  \lambda = \omega^{2}(\ee_{1},\ee_{2}) = (\bxi\cdot\nn_{0})\tr \bF_{\varphi} - \bF_{\varphi}\bxi\cdot\nn_{0} = \left((\tr \bF_{\varphi})\bxi - \bF_{\varphi} \bxi\right) \cdot \nn_{0}.
\end{equation*}
We can therefore rewrite~\eqref{eq:KD-pressure-potential} as~\eqref{eq:LP-Beatty}.

\subsection{An alternative pressure potential}
\label{subsec:alternative-potential}

A Lagrangian formulation (together with its non-holonomic constraints) for the pressure boundary conditions is not unique. There exist in the literature alternative formulations~\cite{Bal1976/77,Cia1988,PC1991}, valid under the stronger condition $\delta \bxi=0$ on $\partial \Sigma_{0}$ (which implies~\eqref{eq:non-holonomic-condition-ref}). Such an alternative pressure potential has been suggested in~\cite[Eq. (1.36)]{Bal1976/77} or \cite[Theorem 2.7-1]{Cia1988}). It can be written as
\begin{equation}\label{eq:B-pressure-potential}
  \tilde{\mathcal{L}}^{P}(\varphi) := \frac{P}{3} \int_{\Sigma_{0}} J_{\varphi} {\bF_{\varphi}}^{-1}\varphi \cdot\nn \, da_{0},
\end{equation}
and differs from $\mathcal{L}^{P}$ given by~\eqref{eq:LP-Beatty} (and deduced from~\eqref{eq:KD-pressure-potential}). On the body, this potential recasts as
\begin{equation*}
  \tilde{\mathcal{L}}^{P}(\pp) := \frac{P}{3} \int_{\Sigma_{\body}} \vol_{\bq}(\pp,\bF\cdot,\bF\cdot),
\end{equation*}
where $\pp=\varphi \circ \pp_{0}$, $\bF=\bF_\varphi \bF_{0}$ and $\Sigma_{\body}=\pp_{0}(\Sigma_{0})$. Its variation can be derived the same way as in the proof of theorem~\ref{thm:delta-LP}, and we get
\begin{equation}\label{eq:derivative-L-B}
  \delta \tilde{\mathcal{L}}^{P}(\delta\pp) = \mathcal{W}^{P}(\delta\pp) + \frac{P}{6}\oint_{\partial \Sigma_{\body}} (\pp \times \delta\pp)\cdot 2\bF d\vec \ell_{\body}.
\end{equation}
Thus, the non-holonomic constraints for $\tilde{\mathcal{L}}^{P}$ to be a pressure potential can be written as
\begin{equation*}
  \oint_{\partial \Sigma_{\body}} (\pp \times \delta\pp)\cdot \bF d\vec \ell_{\body}=0.
\end{equation*}
In an odd way, these constraints depend on the embedding $\pp$ itself, whereas~\eqref{eq:non-holonomic-condition-body} depends instead, in a fully mechanistic manner, on the displacement $\bxi=\pp-\pp_{0}$.

\begin{rem}
  The two Lagrangians $\mathcal{L}^{P}$ and $\tilde{\mathcal{L}}^{P}$ are both valid for mechanical problems for which $\delta\bxi = \delta\pp = 0$ on the boundary $\partial \Sigma_{\body}$ or if $\partial \Sigma_{\body} = \emptyset$. But, contrary to ${\mathcal{L}}^{P}$, $\tilde{\mathcal{L}}^{P}$ is not a pressure potential anymore when the virtual displacement $\delta \bxi$ only remains collinear to the displacement $\bxi=\pp-\pp_{0}$ all along the contour $\partial \Sigma_{\body}$.
\end{rem}

We can calculate explicitly the difference between the variations of $\mathcal{L}^{P}$ and $\tilde{\mathcal{L}}^{P}$. If we set,
\begin{equation*}
  \delta \mathcal{L}^{P} = \frac{1}{6} \int_{\Sigma_{\body}} \delta \omega, \quad \text{and} \quad \delta \tilde{\mathcal{L}}^{P} = \frac{1}{6} \int_{\Sigma_{\body}} \delta \tilde{\omega},
\end{equation*}
then, we have
\begin{equation*}
  \delta \omega - \delta \tilde{\omega} = d\beta,
\end{equation*}
where $\beta$ is the following one-form on $\partial\body$,
\begin{equation*}
  \beta(A) := \vol_{\bq}(\pp,\delta\pp,\Lie_{A}\pp_{0}) - \vol_{\bq}(\pp_{0},\delta\pp,\Lie_{A}(2\pp+\pp_{0})).
\end{equation*}
By Stokes--Ampère formula, the two variations $\delta\mathcal{L}^{P}(\delta\pp)$ and $\delta\tilde{\mathcal{L}}^{P}(\delta\pp)$ of pressure potentials differ then by the contour integral
\begin{align*}
  \frac{1}{6} \oint_{\partial \Sigma_{\body}} \beta & = \frac{1}{6} \oint_{\partial \Sigma_{\body}}
  \vol_{\bq}(\pp,\delta\pp,\bF_{0}\cdot) - \vol_{\bq}(\pp_{0},\delta\pp,(2\bF+\bF_{0})\cdot))
  \\
                                                    & = \frac{1}{6} \oint_{\partial \Sigma_{\body}}
  (\pp\times \delta\pp)\cdot \bF_{0} d \vec \ell_{\body} -
  \frac{1}{6} \oint_{\partial \Sigma_{\body}}  (\pp_{0}\times \delta\pp)\cdot (2\bF+\bF_{0})d \vec \ell_{\body}.
\end{align*}
Therefore, and as expected, the equality $\delta \mathcal{L}^{P} = \delta \tilde{\mathcal{L}}^{P}$ holds when $\delta\pp=\delta\bxi=0$ on $\Sigma_{\body}$ or if $\partial\Sigma_{\body} = \emptyset$.

\section{Conclusion}

We have formulated hyper-elasticity as a variational problem directly on the body $\body$, a three-dimensional compact and orientable manifold with boundary (equipped with a mass measure, assumed to be a non vanishing 3-form), and not necessarily embedded as a reference configuration in space. Accordingly, we have formulated the dead load and pressure types boundary conditions on $\partial \body$. Concerning prescribed pressure, the Poincaré lemma (extended to infinite dimension) has allowed us to obtain, in a straightforward manner, both the pressure potential and optimal non-holonomic constraints for such a potential to exist. The proposed methodology is based on the interpretation of virtual powers as one-forms on the configuration space $\Emb(\body, \espace)$. It is general and can be applied to many others situations. This has allowed us to derive in a systematic way Lagrangians (potentials) in continuum mechanics, and if necessary, to formulate non-holonomic constraints for such potentials to exist.

Finally, we have chosen to consider smooth embeddings rather than $C^p$-embeddings (for instance as in \cite{Seg1986}). Indeed, this choice leads to a very general definition of virtual powers as tensor-distributions (a concept introduced by Lichnerowicz~\cite{Lic1994}).

\appendix

\section{Pullback, pushforward and Lie derivative}
\label{sec:pullback}

The fundamental concept in differential geometry that allows to pass from spatial variables defined on the deformed/actual configuration $\Omega_{\pp}$, to material variables, defined on the body $\body$ or on the reference configuration $\Omega_{0}$ (and vice versa) are the operations \emph{pullback} and \emph{pushforward} (see~\cite[4.7 p. 68]{MH1994} or \cite{SH1997} or~\cite[Chapter V]{Lan1999}, for instance). More precisely, given a diffeomorphism $\phi:M \to N$ between two differentiable manifolds, the pullback $\phi^{*}$ transforms a tensor field $\bt$ defined on $N$ into a tensor field $\phi^{*}\bt$ defined on $M$, while the pushforward $\phi_{*}$ transforms a tensor field $\bT$ defined on $M$ into a tensor field $\phi_{*}\bT$ defined on $N$. The notion of pullback and push forward naturally extend to the case where $\phi$ is an embedding.

\begin{exam}\label{exam:pullbacks-order-two} Two usual examples of pullback/pushforward by $\pp\in \Emb(\body, \espace)$ are
  \begin{equation*}
    \bT=\pp^*\bt=\bF^\star (\bt\circ \pp) \bF,
    \qquad
    \bt=\pp_*\bT=\bF^{-\star} \bT\, \bF^{-1} \circ \pp^{-1},
  \end{equation*}
  when $\bt$ and $\bT$ are second-order covariant tensor fields, and
  \begin{equation*}
    \bT=\pp^*\bt=\bF^{-1} (\bt\circ \pp) \bF^{-\star},
    \qquad
    \bt=\pp_*\bT=\bF \bT\, \bF^{\star} \circ \pp^{-1},
  \end{equation*}
  when $\bt$ and $\bT$ are second-order contravariant tensor fields. Here, $\bF=T\pp: T\body \to T\Omega_{\pp}$ is the linear tangent mapping of $\pp$ and $\bF^{\star}: T^{\star}\Omega_{\pp} \to T^{\star}\body $ is its transpose.
\end{exam}

\begin{rem}\label{rem:com-contractions}
  Pullback and pushforward operations are inverse to each other, meaning that $\phi^{*} = (\phi_{*})^{-1} = (\phi^{-1})_{*}$. They commute moreover with any contraction between covariant and contravariant indices.
\end{rem}

The \emph{Lie derivative} is the infinitesimal version of the pullback. Indeed, let $\uu$ be a vector field on $M$, $\varphi(t)$ be its flow, and let $\bt$ be a tensor field on $M$. The \emph{Lie derivative} of $\bt$ with respect to $\uu$, noted $\Lie_{\uu} \bt$ is defined as
\begin{equation*}
  \Lie_{\uu} \bt := \left.\frac{\partial}{\partial t}\right|_{t=0} \varphi(t)^{*} \bt .
\end{equation*}
When $\bt := \vv$ is a vector field, $\Lie_{\uu} \vv$ is just the \emph{Lie bracket} $[\uu,\vv]=-[\vv,\uu]$ of $\uu$ and $\vv$, and we have moreover
\begin{equation}\label{eq:Lie-derivative-property}
  \Lie_{[\uu,\vv]} \bt = \Lie_{\uu} \Lie_{\vv} \bt - \Lie_{\vv} \Lie_{\uu} \bt .
\end{equation}

The Lie derivative extends without difficulty to \emph{time-dependent vector fields $\uu(t)$}~\cite[Section 1.6]{MH1994}. In that case, the flow $\varphi(t,s)$ of $\uu(t)$ depends on two parameters and one defines
\begin{equation*}
  \Lie_{\uu(s)} \bt := \left.\frac{\partial}{\partial t}\right|_{t=s} \varphi(t,s)^{*} \bt .
\end{equation*}

The following result extends the property $ \partial_{t}\varphi(t)^{*} \bt = \varphi(t)^{*} \Lie_{\uu} \bt$ (see~\cite[Section 1.6]{MH1994}), when the path of diffeomorphisms $\varphi(t)$ is replaced by a path of embeddings $\pp(t): \body \to \espace$.

\begin{lem}\label{lem:magic-formula}
  Let $\pp(t)$ be a path of embeddings, $\uu = (\partial_{t}\pp)\circ \pp^{-1}$ be its (right) Eulerian velocity and $\bt$ be a tensor field defined along $\pp(t)$ (\textit{i.e.}, on $\Omega_{p(t)} = \pp(t)(\body)$ and possibly time-dependent). Then
  \begin{equation*}
    \partial_{t}(\pp^{*}\bt) = \pp^{*}\left( \partial_{t}\bt + \Lie_{\uu}\bt \right).
  \end{equation*}
\end{lem}

\section{Differential forms}
\label{sec:differential-forms}

For basic materials on differential forms one may look at~\cite{Car1970,AMR1988,Lan1999}. Let $\Omega^{k}(M)$ be the space of \emph{differentials forms} of degree $k$, \textit{i.e.}, covariant tensor fields $\omega$ on the manifold $M$ of order $k$, which are alternate. The \emph{exterior derivative} is a differential operator of order one
\begin{equation*}
  d : \Omega^{k}(M) \to \Omega^{k+1}(M)
\end{equation*}
which extends the differential of a function to differential forms of any degree. For instance, in any local coordinate system $(x^{i})$, we have
\begin{equation*}
  (d\alpha)_{ij} = \partial_{i} \alpha_{j} - \partial_{j} \alpha_{i},
\end{equation*}
for a $1$-form $\alpha$ and
\begin{equation*}
  (d\omega)_{ijk} = \partial_{i} \alpha_{jk} - \partial_{j} \alpha_{ik} + \partial_{k} \alpha_{ij},
\end{equation*}
for a $2$-form $\alpha$. Given a vector field $A$ on $M$, the \emph{inner product}
\begin{equation*}
  i_{A} : \Omega^{k}(M) \to \Omega^{k-1}(M)
\end{equation*}
is defined by
\begin{equation*}
  (i_{A} \alpha)_{m}(B_{1}, \dotsc , B_{k-1}) := \alpha_{m}(A(m),B_{1}, \dotsc , B_{k-1}),
\end{equation*}
for any $B_{1}, \dotsc , B_{k-1} \in T_{m}M$, if $k \ge 1$ and $i_{A}\alpha = 0$ if $k=0$.

These two linear operators are related to each other and to the Lie derivative $\Lie_{A}$ by \emph{Cartan's magic Formula}
\begin{equation}\label{eq:Cartan-formula}
  \Lie_{A} = d \circ i_{A} + i_{A} \circ d.
\end{equation}

The following lemma is useful for our computations, in which $\vol_{\bq}$ is the Riemannian volume for Euclidean metric $\bq$ (see \autoref{sec:volume-forms}).

\begin{lem}\label{lem:alpha-exterior-derivative}
  Let $\ff_{1}, \ff_{2}, \ff_{3}$ be three \emph{vector valued functions}, defined on a manifold $M$ and let
  \begin{equation*}
    \alpha(A) = \vol_{\bq}(\ff_{1}, \ff_{2}, \Lie_{A} \ff_{3}), \qquad A\in T_{m} M.
  \end{equation*}
  Then
  \begin{equation}\label{eq:dalphaABLXp}
    \begin{aligned}
      d\alpha(A,B) & = \vol_{\bq}(\Lie_{A}\ff_{1}, \ff_{2} , \Lie_{B}\ff_{3}) - \vol_{\bq}(\Lie_{B}\ff_{1}, \ff_{2} , \Lie_{A}\ff_{3})
      \\
                   & \quad + \vol_{\bq}(\ff_{1}, \Lie_{A}\ff_{2} , \Lie_{B}\ff_{3}) - \vol_{\bq}(\ff_{1}, \Lie_{B}\ff_{2} , \Lie_{A}\ff_{3}).
    \end{aligned}
  \end{equation}
  with $A, B\in T_{m}M$.
\end{lem}

\begin{proof}
  By Cartan's formula~\eqref{eq:Cartan-formula}, we have
  \begin{equation*}
    d\alpha(A,B) = \Lie_{A} \alpha(B) - \Lie_{B} \alpha(A) - \alpha([A,B]),
  \end{equation*}
  where $[A,B]:=\Lie_{A} B$ is the Lie Bracket. But
  \begin{equation*}
    \Lie_{A} \alpha(B) = \vol_{\bq}(\Lie_{A}\ff_{1}, \ff_{2}, \Lie_{B}\ff_{3}) + \vol_{\bq}(\ff_{1}, \Lie_{A}\ff_{2}, \Lie_{B}\ff_{3}) + \vol_{\bq}(\ff_{1}, \ff_{2}, \Lie_{A} \Lie_{B}\ff_{3}),
  \end{equation*}
  \begin{equation*}
    \Lie_{B} \alpha(A) = \vol_{\bq}(\Lie_{B}\ff_{1}, \ff_{2}, \Lie_{A}\ff_{3}) + \vol_{\bq}(\ff_{1}, \Lie_{B}\ff_{2}, \Lie_{A}\ff_{3}) + \vol_{\bq}(\ff_{1}, \ff_{2}, \Lie_{B} \Lie_{A}\ff_{3}),
  \end{equation*}
  and
  \begin{equation*}
    \alpha([A,B]) = \vol_{\bq}(\ff_{1}, \ff_{2}, \Lie_{[A,B]}\ff_{3}).
  \end{equation*}
  Now, by~\eqref{eq:Lie-derivative-property}, we have
  \begin{equation*}
    \Lie_{A} \Lie_{B}\ff_{3} - \Lie_{B} \Lie_{A}\ff_{3} = \Lie_{[A,B]}\ff_{3}
  \end{equation*}
  and we get thus~\eqref{eq:dalphaABLXp}.
\end{proof}

A differential form $\alpha \in \Omega^{k}(M)$ is said to be closed if $d\alpha = 0$. It is said to be exact if it can be written as $d\beta$ with $\beta \in \Omega^{k-1}(M)$. Since $d \circ d = 0$, \emph{any exact form is closed}. The following result, attributed to Poincaré, implies that \emph{any closed form defined on an open convex set of $\RR^{n}$ is exact.}

\begin{lem}[Poincaré lemma]\label{lem:Poincare}
  Let $U\subset \RR^{n}$ be a convex open set and $\alpha \in \Omega^{k}(U)$ ($1 \le k \le n$). If $d\alpha = 0$, then there exists $\beta \in \Omega^{k-1}(U)$ such that $\alpha = d\beta$.
\end{lem}

The important point is that the proof of this lemma is constructive. A primitive $\beta$ of a closed differential form $\alpha$ defined on $U$ is explicitly constructed. Indeed, let $\xx_{0} \in U$ and let $\phi^{t}(\xx) = e^{t}\xx + (1-e^{t})\xx_{0}$ be the flow of the radial vector field $X(\xx) = \xx - \xx_{0}$ defined on $U$. Define the linear operator $K: \Omega^{k}(U) \to \Omega^{k-1}(U)$ by
\begin{equation}\label{eq:Poincare-formula}
  K\alpha = \int_{-\infty}^{0} (\phi^{t})^{*} i_{X}\alpha \, dt.
\end{equation}
Then, using Cartan's formula~\eqref{eq:Cartan-formula} and the fact that pullbacks and exterior derivative commute, we get
\begin{equation*}
  d(\phi^{t})^{*} i_{X}\alpha = (\phi^{t})^{*} di_{X}\alpha = (\phi^{t})^{*} (\Lie_{X}\alpha - i_{X}d\alpha) = (\phi^{t})^{*}\Lie_{X}\alpha,
\end{equation*}
and thus
\begin{equation*}
  dK\alpha = \int_{-\infty}^{0} d(\phi^{t})^{*} i_{X}\alpha \, dt = \int_{-\infty}^{0} (\phi^{t})^{*} \Lie_{X}\alpha \, dt = \int_{-\infty}^{0} \frac{d}{dt}((\phi^{t})^{*} \alpha ) \, dt = \alpha .
\end{equation*}
Therefore, $\beta := K\alpha$ is a primitive of $\alpha$.

\begin{rem}
  When the form $\alpha$ is not closed we have
  \begin{equation}\label{eq:Poincare-formula-full}
    dK\alpha = \alpha - \int_{-\infty}^{0} (\phi^{t})^{*} i_{X}d\alpha \, dt .
  \end{equation}
  This last integral can be used to formulate non-holonomic constraints for a primitive to exist, even when $\alpha$ is not closed.
\end{rem}

\section{Volume forms and area elements}
\label{sec:volume-forms}

A \emph{volume form} on a manifold $M$ of dimension $n$ is a differential form of degree $n$ that does not vanish at any point. A manifold $M$ which has a volume form is necessarily \emph{orientable}~\cite[Section 1.G]{GHL2004}. On an orientable Riemannian manifold $(M,g)$, there is a unique volume form, noted $\vol_{g}$, and called the \emph{Riemannian volume form} such that
\begin{equation*}
  (\vol_{g})_{m}(\ee_{1}, \dotsc , \ee_{n}) = 1,
\end{equation*}
for any direct orthonormal basis of $T_{m}M$ and any point $m \in M$. In a local coordinate system $(x^{i})$, this volume form can be written as
\begin{equation}\label{eq:riemannian-volume-local-expression}
  \vol_{g} = \sqrt{\det (g_{ij})} \, dx^{1} \wedge \dotsb \wedge dx^{n}.
\end{equation}

Let us now consider an orientable manifold with boundary $M$. Let $\omega$ be a volume form on $M$ and $x \in \partial M$. In a local chart in the vicinity of $x$, we can consider the form $i_{\nn}\omega$ where $\nn$ is a vector with a $x_{1}$ component \emph{strictly positive}. The class of this volume form defines the induced orientation on $\partial M$, known as the convention of the \emph{outer normal}. For example, consider the prototype manifold $]-\infty,0] \times \RR^{n-1}$ with the orientation defined by the volume form  $dx^{1} \wedge \dotsb \wedge dx^{n}$. Then the orientation induced on the boundary $\RR^{n-1}$ of $]-\infty,0] \times \RR^{n-1}$ is represented by the volume form $dx^{2} \wedge \dotsb \wedge dx^{n}$.

Let $(M,g)$ be an oriented $3$-dimensional Riemannian manifold with boundary $\partial M$. Then, the Riemannian metric $g$ on $M$ induces by restriction a Riemannian metric on $\partial M$. Let $\vol_{g}$ be the Riemannian volume form on $M$ and $\nn$ be the outer unit normal on the boundary $\partial M$. Then, one can show that the Riemannian volume on the $2$-dimensional manifold $\partial M$ is the $2$-form
\begin{equation*}
  da := i_{\nn}\vol_{g},
\end{equation*}
which is called the \emph{area element} of $\partial M$.

\begin{rem}\label{rem:nds}
  Let $X$ be a vector field defined on $\partial M$ (not necessarily tangent to $\partial M$). Then, as a $2$-form on $\partial M$, we have the following identities
  \begin{equation}\label{eq:nds}
    i_{X}\vol_{g} = \langle X, \nn \rangle i_{\nn}\vol_{g} = \langle X, \nn \rangle da.
  \end{equation}
\end{rem}

The following theorem happens to be extremely useful when deriving boundary conditions on the body.

\begin{thm}\label{thm:generalized-Nanson}
  Let $\pp$, $\pp_{0}$ be two orientation-preserving embeddings from $\body$ to $\espace$. Let $\ww$ be a vector field defined on $\partial\Omega_{\pp}$ (not necessarily tangent to $\partial\Omega_{\pp}$). Then
  \begin{equation}\label{eq:generalized-Nanson}
    \pp^{*}(\rho (\ww \cdot \nn) da) = \rho_{\bgamma_{0}} (\ww\circ\pp \cdot \bq^{-1}\bF^{-\star}\bgamma_{0} \bN_{0}) da_{\bgamma_{0}},
  \end{equation}
  where $\bF^{\star}: T^{\star}\espace \to T^{\star}\body $ is the (metric free) transpose
  of $\bF=T\pp: T\body \to T\espace $, $\rho_{\bgamma_{0}} = \pp_{0}^{*}\, \rho_{0}=\rho_{0}\circ\pp_{0}$ and $da_{\bgamma_{0}}$ is the area density on $\partial\body$ relative to the metric $\bgamma_{0} = {\pp_{0}}^{*}\bq$.
\end{thm}

\begin{proof}
  Using~\eqref{eq:nds}, and the fact that $\pp^{*}\rho = \rho_{\bgamma}$, we have
  \begin{equation*}
    \pp^{*}(\rho (\ww \cdot \nn) da) = \rho_{\bgamma} \pp^{*}( (\ww \cdot \nn) da) = \rho_{\bgamma} \pp^{*}( i_{\ww} \vol_{\bq}) = \rho_{\bgamma} i_{\pp^{*}\ww} \vol_{\bgamma},
  \end{equation*}
  because $\pp$ is an orientation-preserving Riemannian isometry. Now, mass conservation leads (see~\autoref{sec:kinematics}) to
  \begin{equation*}
    \mu = \rho_{\bgamma}\vol_{\bgamma} = \rho_{\bgamma_{0}}\vol_{\bgamma_{0}}.
  \end{equation*}
  Therefore, we deduce that
  \begin{equation*}
    \pp^{*}(\rho (\ww \cdot \nn) da) = \rho_{\bgamma_{0}} i_{\pp^{*}\ww} \vol_{\bgamma_{0}} = \rho_{\bgamma_{0}} \langle \pp^{*}\ww, \bN_{0} \rangle_{\gamma_{0}} da_{\bgamma_{0}}.
  \end{equation*}
  But
  \begin{equation*}
    \langle \pp^{*}\ww, \bN_{0} \rangle_{\gamma_{0}} = \left( \bgamma_{0}\bN_{0}, \bF^{-1} \ww \circ \pp \right) = \left( \bF^{-\star}\bgamma_{0}\bN_{0}, \ww \circ \pp \right) = (\bq^{-1}\bF^{-\star}\bgamma_{0}\bN_{0}) \cdot \ww \circ \pp,
  \end{equation*}
  where $(\cdot,\cdot)$ is the duality bracket. This achieves the proof.
\end{proof}

\begin{rem}\label{rem:Nanson}
  If we identify the body $\body$ with a reference configuration $\Omega_{0}$, and set thus $\pp_{0}\equiv \id$, $\bgamma_{0}\equiv \bq$, $\pp\equiv \varphi$ and $\bF\equiv\bF_{\varphi}$, then, \eqref{eq:generalized-Nanson} can be written as
  \begin{equation*}
    \varphi^{*}\big(\rho(\ww \cdot \nn)\, da \big) = \rho_{0} (\ww\circ\varphi \cdot \bq^{-1}\bF_{\varphi}^{-\star}\bq \nn_{0}) da_{0} = \rho_{0} (\ww\circ\varphi \cdot \bF_{\varphi}^{-t} \nn_{0}) da_{0},
  \end{equation*}
  which we can recast as
  \begin{equation*}
    \varphi^{*}\big((\ww \cdot \nn)\, da \big) = \frac{\rho_{0}}{\varphi^{*}\rho} (\ww\circ\varphi \cdot \bF_{\varphi}^{-t} \nn_{0}) da_{0} = J_{\varphi} (\ww\circ\varphi \cdot \bF_{\varphi}^{-t} \nn_{0}) da_{0},
  \end{equation*}
  because $\rho_{0}/\varphi^{*}\rho = J_{\varphi}$ by~\eqref{eq:mass-conservation}. This equality is known in continuum mechanics as \emph{Nanson's formula}~\cite{Nan1878} and often written in condensed form as
  \begin{equation}\label{eq:Nanson-formula}
    \nn \, da = J_{\varphi} \bF_{\varphi}^{-t} \nn_{0} \, da_{0}.
  \end{equation}
\end{rem}

\section{Fréchet topology on the space of embeddings}
\label{sec:Frechet-topology}

In this paper, we are interested into the set $\Emb(\body,\espace)$ of smooth embeddings, which is a subset of $\Cinf(\body,\espace)$. Once an origin of $\espace$ has been chosen, the affine space $\espace$ inherits the structure of a vector space which is isomorphic to its translation space $E$. Thus, $\Cinf(\body,\espace)$ can be considered as the vector space of smooth vector valued functions with values in the vector space $\espace$. This vector space $\Cinf(\body,\espace)$ is not a Banach space; its topology is not defined by a norm but by a countable family of \emph{semi-norms}. These semi-norms can be described either by choosing a Riemannian metric on $\body$ or by choosing an embedding $\pp_{0}$ of $\body$ into $\espace$. In both cases, one can prove anyway that \emph{the defined topology does not depend of the particular choice of the metric or of the embedding}, used to build this topology.

Here, we chose to describe this topology using an embedding $\pp_{0}$. This means that we describe first a topology on $\Cinf(\Omega_{0},\espace)$, where $\Omega_{0} = \pp_{0}(\body)$, and then a topology on $\Cinf(\body,\espace)$, using the invertible linear mapping
\begin{equation*}
  L_{\pp_{0}}: \Cinf(\body,\espace) \to \Cinf(\Omega_{0},\espace), \qquad \pp \mapsto \varphi := \pp \circ \pp_{0}^{-1}.
\end{equation*}
Thus, a subset $U \subset \Cinf(\body,\espace)$ is open if and only if $L_{\pp_{0}}(U)$ is open and we define the semi-norm of $\pp \in \Cinf(\body,\espace)$ as the semi-norm of $\pp \circ \pp_{0}^{-1} \in \Cinf(\Omega_{0},\espace)$. The topology on $\Cinf(\Omega_{0},\espace)$ is defined by the family of $C^{k}$-semi-norms
\begin{equation*}
  \norm{f}_{k} := \max_{k_{1} + k_{2} + k_{3} = k} \left\{\sup_{x \in \Omega_{0}} \frac{\partial^{k}f}{\partial_{x}^{k_{1}}\partial_{y}^{k_{2}}\partial_{z}^{k_{3}}}(x)\right\}.
\end{equation*}
For this topology, the \emph{semi-balls}
\begin{equation*}
  B_{k}(f_{0},r) := \set{f \in \Cinf(\Omega_{0},\espace); \norm{f-f_{0}}_{k} < r}
\end{equation*}
are always open sets. The space $\Cinf(\body,\espace)$ is an example of a so-called \emph{Fréchet space}~\cite{Ham1982,Rud1991}. These spaces are not nice from the point of view of Analysis since essential results such as \emph{Inverse mapping theorem} or \emph{local existence of solutions} for ordinary differential equations are no longer true without hard-to-check additional hypotheses~\cite{Ham1982}. There are however other choices of topological and diffeological structures on this space and the interested reader may consult~\cite{KM1997}. The set $\Emb(\body,\espace)$ of smooth embeddings from $\body$ to $\espace$ can be shown to be an open set of the vector space $\Cinf(\body,\espace)$ as a corollary of the following lemma (see also~\cite{Hir1976,Mil1984,KOS2017,SE2020}).

\begin{lem}\label{lem:convex-neighborhood}
  Let $\pp_{0} \in \Emb(\body,\espace)$. The neighborhood of $\pp_{0}$, defined by
  \begin{equation*}
    \mathcal{U}_{\pp_{0}} := \set{\pp \in \Cinf(\body,\espace); \; \norm{\pp - \pp_{0}}_{1} < 1}
  \end{equation*}
  is an \emph{open convex set} of $\Cinf(\body,\espace)$, contained in $\Emb(\body,\espace)$.
\end{lem}

\begin{proof}
  Note first that
  \begin{equation*}
    L_{\pp_{0}}(\mathcal{U}_{\pp_{0}}) := \set{\varphi \in \Cinf(\Omega_{0},\espace); \; \norm{\varphi - \id}_{1} < 1}
  \end{equation*}
  is an open convex subset of $\Cinf(\Omega_{0},\espace)$. By the very definition of the topology on $\Cinf(\body,\espace)$ and the fact that $L_{\pp_{0}}$ is linear, we deduce that $\mathcal{U}_{\pp_{0}}$ is an open convex subset of $\Cinf(\body,\espace)$. It remains to show that each $\pp \in \mathcal{U}_{\pp_{0}}$ is an embedding. We will show that each $\varphi \in L_{\pp_{0}}(\mathcal{U}_{\pp_{0}})$ is an embedding and the conclusion follows straightforwardly. We will prove first that each vector valued function $\varphi \in L_{\pp_{0}}(\mathcal{U}_{\pp_{0}})$ is injective. To do so, let $\varphi \in \mathcal{U}_{\pp_{0}}$ and suppose that $\pp(\xx)=\pp(\yy)$, we get then, thanks to the mean value theorem, that
  \begin{equation*}
    \norm{\xx-\yy} = \norm{\xx-\varphi(\xx)+\varphi(\yy)-\yy} = \norm{(\id-\varphi)(\xx)-(\id-\varphi)(\yy)} \le \norm{\varphi - \id}_{1} \norm{\xx-\yy}
  \end{equation*}
  and thus that $\xx=\yy$, because $\norm{\varphi - \id}_{1} < 1$. Next, we will show that $\varphi$ is an immersion, \textit{i.e}, that $\bF_{\varphi} = T\varphi$ is injective. Thus, assume that $\bF_{\varphi}.\delta\xx = \bF_{\varphi}.\delta\yy$, then,
  \begin{equation*}
    \norm{\delta\xx-\delta\yy} = \norm{(\Id-\bF_{\varphi})\delta\xx -(\Id-\bF_{\varphi})\delta\yy} \le \norm{\bF_{\varphi}-\Id}_{0}\norm{\delta\xx-\delta\yy} = \norm{\varphi - \id}_{1} \norm{\delta\xx-\delta\yy},
  \end{equation*}
  and therefore $\delta\xx = \delta\yy$, because $\norm{\varphi - \id}_{1} < 1$. To finish, we observe that an injective immersion from a compact manifold into space is always an embedding. Since it is assumed in solid mechanics that $\body$ (and thus $\Omega_{0}$) is compact, this achieves the proof.
\end{proof}

\begin{rem}
  The set $\Met(\body)$ of all Riemannian metrics on $\body$ is a subset of the vector space of smooth sections $\Gamma(S^{2}T^{\star}\body)$ of the vector bundle $S^{2}T^{\star}\body$. A Fréchet topology on $\Gamma(S^{2}T^{\star}\body)$ can be constructed, using the same arguments used to build one on $\Cinf(\body,\espace)$. The set $\Met(\body)$ can be shown to be an open convex set of $\Gamma(S^{2}T^{\star}\body)$.
\end{rem}

\begin{rem}
  In this paper, we have chosen to work in the smooth category. It is however possible to work in the category $C^{p}$, where $p \ge 1$, as it was done for instance in~\cite{Seg1986}.
\end{rem}

\section{The first Piola--Kirchhoff tensor on the body}
\label{sec:PK1}

When so-called dead loads are involved, \textit{i.e.}, loads per unit area of direction and intensity independent of the deformation of the medium, the associated Neumann condition corresponds to the prescribed stress traction vector $\vec t_{0}$ on the boundary part $\Sigma_{0}^{(DL)}\subset \partial \Omega_{0}$ (with $\Omega_{0}=\pp_{0}(\body)$ the reference configuration) but with values in space $\espace$,
\begin{equation}\label{eq:CLDL}
  \left.\widehat \bP \nn_{0}\right|_{\Sigma_{0}^{(DL)}} = \vec t_{0},
\end{equation}
where $\widehat \bP$ is the (mixed, two point \cite{MH1994,Hau2002}) first Piola--Kirchhoff tensor,
\begin{equation*}
  \widehat \bP := \rho_{0} \bF_{\varphi}\bS\, \bq =\rho_{0} (\btau \circ \varphi) \bF_\varphi^{-\star}\bq,
\end{equation*}
such that
\begin{equation*}
  \int_{\partial \Omega_{\pp}} (\widehat{\bsigma}\nn \cdot \ww) \, da =
  \int_{\partial \Omega_{0}} (\widehat{\bP}\nn_{0} \cdot \delta \varphi ) \, da_{0},
\end{equation*}
for any virtual velocity $\ww =\delta \varphi \circ \varphi^{-1}$.
Recall that $\bS=\varphi^{*} \btau =\varphi^{*}(\bsigma/\rho)$ is the second Piola--Kirchhoff stress tensor (defined on $\Omega_{0}$). We have then the following result.

\begin{lem}\label{lem:PK-body}
  We have
  \begin{equation*}
    \int_{\partial \Omega_{\pp}} (\widehat{\bsigma}\nn \cdot \ww) \, da =
    \int_{\partial \body} (\widehat{\mpi}\bN_{0} \cdot \delta \pp ) \, da_{\bgamma_{0}},
  \end{equation*}
  for any virtual velocity $\ww=\delta \pp \circ \pp^{-1}$, where
  \begin{equation*}
    \widehat{\mpi} := \rho_{\bgamma_{0}} (\btau \circ \pp) \bF^{-\star}\bq = \rho_{\bgamma_{0}} \bF \btheta \bgamma_{0},
  \end{equation*}
  is the mixed stress tensor defined on $\body$ but with values on Euclidean space $\espace$, with $\bgamma_{0}=\pp_{0}^*\, \bq$, $\rho_{\bgamma_{0}}= \pp_{0}^{*}\, \rho_{0}=\rho_{0}\circ \pp_{0}$, $\btau := \bsigma / \rho$ and $\btheta = \pp^{*}\btau$.
\end{lem}

\begin{rem}\label{rem:DL-on-the-body}
  Thanks to this result, the dead load boundary condition \eqref{eq:CLDL} rewrites on the body
  \begin{equation*}
    \left.\widehat{\bpi} \bN_{0}\right|_{\Sigma_\body^{(DL)}} = \vec t_{0} \circ \pp_{0},
    \qquad
    \Sigma_\body^{(DL)}=\pp_{0}^{-1}( \Sigma_{0}^{(DL)}).
  \end{equation*}
\end{rem}

\begin{proof}[Proof of lemma~\ref{lem:PK-body}]
  By the symmetry of $\bsigma$ and the change of variable formula, we get
  \begin{equation*}
    \int_{\partial \Omega_{\pp}} (\widehat{\bsigma}\nn \cdot \ww) \, da = \int_{\partial \Omega_{\pp}} (\nn \cdot \widehat{\bsigma}\ww) \, da = \int_{\partial \body} \pp^{*}((\nn \cdot \widehat{\bsigma}\ww)\, da) = \int_{\partial \body} \pp^{*}(\rho(\nn \cdot \widehat{\btau}\ww)\, da).
  \end{equation*}
  Therefore, using theorem~\ref{thm:generalized-Nanson}, we get
  \begin{equation*}
    \pp^{*}((\widehat{\bsigma}\nn \cdot \ww)\, da) = \pp^{*}(\rho(\nn \cdot \widehat{\btau}\ww)\, da) = \rho_{\bgamma_{0}} ((\widehat{\btau}\ww)\circ\pp \cdot \bq^{-1}\bF^{-\star}\bgamma_{0} \bN_{0}) da_{\bgamma_{0}}.
  \end{equation*}
  But
  \begin{equation*}
    (\widehat{\btau}\ww)\circ\pp \cdot \bq^{-1}\bF^{-\star}\bgamma_{0} \bN_{0} = \ww\circ\pp \cdot (\widehat{\btau}\circ\pp)\bq^{-1}\bF^{-\star}\bgamma_{0} \bN_{0} = \ww\circ\pp \cdot (\btau\circ\pp)\bF^{-\star}\bgamma_{0} \bN_{0},
  \end{equation*}
  and $\btau\circ\pp = \bF\btheta\bF^{\star}$, by the very definition of the push forward $\btau = \pp_{*}\btheta$ (see~\ref{exam:pullbacks-order-two}).
  We get thus finally
  \begin{equation*}
    \int_{\partial \Omega_{\pp}} (\widehat{\bsigma}\nn \cdot \ww) \, da = \int_{\partial \body}  \delta\pp \cdot \widehat{\bpi}\bN_{0}\, da_{\bgamma_{0}},
  \end{equation*}
  where
  \begin{equation*}
    \widehat{\bpi} := \rho_{\bgamma_{0}}(\btau\circ\pp)\bF^{-\star}\bgamma_{0} = \rho_{\bgamma_{0}}\bF\btheta\bgamma_{0},
  \end{equation*}
  which concludes the proof.
\end{proof}



\begin{thebibliography}{10}

\bibitem{AMR1988}
R.~Abraham, J.~E. Marsden, and T.~Ratiu.
\newblock {\em Manifolds, {T}ensor {A}nalysis, and {A}pplications}, volume~75
  of {\em Applied Mathematical Sciences}.
\newblock Springer-Verlag, New York, second edition, 1988.

\bibitem{Arn1965}
V.~I. Arnold.
\newblock On conditions for non-linear stability of plane stationary
  curvilinear flows of an ideal fluid.
\newblock {\em Dokl. Akad. Nauk SSSR}, 162:975--978, 1965.

\bibitem{Arn1966}
V.~I. Arnold.
\newblock Sur la géométrie différentielle des groupes de {L}ie de
  dimension infinie et ses applications à l'hydrodynamique des fluides
  parfaits.
\newblock {\em Ann. Inst. Fourier (Grenoble)}, 16(fasc. 1):319--361, 1966.

\bibitem{Bal1976/77}
J.~M. Ball.
\newblock Convexity conditions and existence theorems in nonlinear elasticity.
\newblock {\em Arch. Rational Mech. Anal.}, 63(4):337--403, Dec. 1976/77.

\bibitem{Bea1970}
M.~Beatty.
\newblock Stability of hyperelastic bodies subject to hydrostatic loading.
\newblock {\em Non-linear Mech.}, 5:367--383, 1970.

\bibitem{Ber2012}
A.~Bertram.
\newblock {\em Elasticity and Plasticity of Large Deformations}.
\newblock Springer Berlin Heidelberg, third edition, 2012.
\newblock An introduction.

\bibitem{Car1970}
H.~Cartan.
\newblock {\em Differential {F}orms}.
\newblock Translated from the French. Houghton Mifflin Co., Boston, Mass, 1970.

\bibitem{Cia1988}
P.~G. Ciarlet.
\newblock {\em Mathematical {E}lasticity. Vol. I}, volume~20 of {\em Studies in
  Mathematics and its Applications}.
\newblock North-Holland Publishing Co., Amsterdam, 1988.
\newblock Three-dimensional elasticity.

\bibitem{Cla2010}
B.~Clarke.
\newblock The metric geometry of the manifold of {R}iemannian metrics over a
  closed manifold.
\newblock {\em Calc. Var. Partial Differential Equations}, 39(3-4):533--545,
  2010.

\bibitem{Dim2011}
Y.~I. Dimitrienko.
\newblock {\em Nonlinear Continuum Mechanics and Large Inelastic Deformations},
  volume 174 of {\em Solid Mechanics and its Applications}.
\newblock Springer Netherlands, 2011.

\bibitem{DS1999}
S.~Doll and K.~Schweizerhof.
\newblock On the development of volumetric strain energy functions.
\newblock {\em Journal of Applied Mechanics}, 67(1):17--21, Oct. 1999.

\bibitem{Ebi1968}
D.~G. Ebin.
\newblock On the space of {R}iemannian metrics.
\newblock {\em Bull. Amer. Math. Soc.}, 74:1001--1003, 1968.

\bibitem{EM1970a}
D.~G. Ebin and J.~E. Marsden.
\newblock {G}roups of diffeomorphisms and the motion of an incompressible
  fluid.
\newblock {\em Ann. of Math. (2)}, 92:102--163, 1970.

\bibitem{Ein1988}
A.~Einstein.
\newblock {\em The {M}eaning of {R}elativity}.
\newblock Princeton University Press, Princeton, NJ, 1988.

\bibitem{EJdL2019}
M.~Epstein, V.~M. Jiménez, and M.~de~León.
\newblock Material geometry.
\newblock {\em J. Elasticity}, 135(1-2):237--260, 2019.

\bibitem{ES1980}
M.~Epstein and R.~Segev.
\newblock Differentiable manifolds and the principle of virtual work in
  continuum mechanics.
\newblock {\em Journal of Mathematical Physics}, 21(5):1243--1245, May 1980.

\bibitem{Eri1962}
A.~C. Eringen.
\newblock {\em Nonlinear {T}heory of {C}ontinuous {M}edia}.
\newblock McGraw-Hill Book Co., New York-Toronto-London, 1962.

\bibitem{Fia2004}
Z.~Fiala.
\newblock Time derivative obtained by applying the {R}iemannian manifold of
  {R}iemannian metrics to kinematics of continua.
\newblock {\em C. R. Mecanique}, 332:97--102, 2004.

\bibitem{Fia2011}
Z.~Fiala.
\newblock Geometrical setting of solid mechanics.
\newblock {\em Ann. Physics}, 326(8):1983--1997, 2011.

\bibitem{Fia2016}
Z.~Fiala.
\newblock Geometry of finite deformations and time-incremental analysis.
\newblock {\em International Journal of Non-Linear Mechanics}, 81:230--244, May
  2016.

\bibitem{FG1989}
D.~S. Freed and D.~Groisser.
\newblock The basic geometry of the manifold of {R}iemannian metrics and of its
  quotient by the diffeomorphism group.
\newblock {\em Michigan Math. J.}, 36(3):323--344, 1989.

\bibitem{GHL2004}
S.~Gallot, D.~Hulin, and J.~Lafontaine.
\newblock {\em Riemannian Geometry}.
\newblock Universitext. Springer Berlin Heidelberg, Berlin, third edition,
  2004.

\bibitem{GM1991}
O.~Gil-Medrano and P.~W. Michor.
\newblock The {R}iemannian manifold of all {R}iemannian metrics.
\newblock {\em Quart. J. Math. Oxford Ser. (2)}, 42(166):183--202, 1991.

\bibitem{GZ1968}
A.~E. Green and W.~Zerna.
\newblock {\em Theoretical {E}lasticity}.
\newblock Second edition. Clarendon Press, Oxford, 1968.

\bibitem{GLM2014}
N.~Grubic, P.~G. LeFloch, and C.~Mardare.
\newblock The equations of elastostatics in a {R}iemannian manifold.
\newblock {\em Journal de Math{é}matiques Pures et Appliqu{é}es},
  102(6):1121--1163, Dec. 2014.

\bibitem{Ham1982}
R.~S. Hamilton.
\newblock {T}he inverse function theorem of {N}ash and {M}oser.
\newblock {\em Bull. Amer. Math. Soc. (N.S.)}, 7(1):65--222, 1982.

\bibitem{Hart66}
L.~Hart-Smith.
\newblock Elasticity parameters for finite deformations of rubber-like
  materials.
\newblock {\em J. Appl. Phys.}, 17:608--625, 1966.

\bibitem{Hau2002}
P.~Haupt.
\newblock {\em Continuum Mechanics and Theory of Materials.}
\newblock 2\textsuperscript{nd} Edition, Springer, Berlin, 2002.
\newblock Traduction de la quatrième édition allemande par G. Juvet et
  R. Leroy.

\bibitem{Hil1924}
D.~Hilbert.
\newblock Die {G}rundlagen der {P}hysik.
\newblock {\em Math. Ann.}, 92(1-2):1--32, 1924.

\bibitem{Hir1976}
M.~W. Hirsch.
\newblock {\em Differential Topology}.
\newblock Springer New York, 1976.

\bibitem{IKT2013}
H.~Inci, T.~Kappeler, and P.~Topalov.
\newblock {\em On the Regularity of the Composition of Diffeomorphisms}, volume
  226 of {\em Memoirs of the American Mathematical Society}.
\newblock American Mathematical Society, first edition, Mar. 2013.

\bibitem{JdLE2018}
V.~M. Jiménez, M.~de~León, and M.~Epstein.
\newblock Characteristic distribution: an application to material bodies.
\newblock {\em J. Geom. Phys.}, 127:19--31, 2018.

\bibitem{JdLE2020}
V.~M. Jiménez, M.~de~León, and M.~Epstein.
\newblock Material distributions.
\newblock {\em Math. Mech. Solids}, 25(7):1450--1458, 2020.

\bibitem{KM1997}
A.~Kriegl and P.~W. Michor.
\newblock {\em The {C}onvenient {S}etting of {G}lobal {A}nalysis}, volume~53 of
  {\em Mathematical Surveys and Monographs}.
\newblock American Mathematical Society, Providence, RI, 1997.

\bibitem{KOS2017}
R.~Kupferman, E.~Olami, and R.~Segev.
\newblock Continuum dynamics on manifolds: Application to elasticity of
  residually-stressed bodies.
\newblock {\em Journal of Elasticity}, 128(1):61--84, Jan. 2017.

\bibitem{Lan1999}
S.~Lang.
\newblock {\em Fundamentals of Differential Geometry}, volume 191 of {\em
  Graduate Texts in Mathematics}.
\newblock Springer-Verlag, New York, 1999.

\bibitem{Lic1994}
A.~Lichnerowicz.
\newblock Tensor-distributions.
\newblock In {\em Magnetohydrodynamics: Waves and Shock Waves in Curved
  Space-Time}, pages 1--17. Springer Netherlands, 1994.

\bibitem{MH1994}
J.~E. Marsden and T.~J.~R. Hughes.
\newblock {\em Mathematical {F}oundations of {E}lasticity}.
\newblock Dover Publications, Inc., New York, 1994.
\newblock Corrected reprint of the 1983 original.

\bibitem{Mil1984}
J.~Milnor.
\newblock {R}emarks on infinite-dimensional {L}ie groups.
\newblock In {\em Relativity, groups and topology, II (Les Houches, 1983)},
  pages 1007--1057. North-Holland, Amsterdam, 1984.

\bibitem{Nan1878}
E.~J. Nanson.
\newblock Note on hydrodynamics.
\newblock {\em Messenger of Mathematics}, 7:182--185, 1878.

\bibitem{Nol1959}
W.~Noll.
\newblock The {F}oundations of {C}lassical {M}echanics in the {L}ight of
  {R}ecent {A}dvances in {C}ontinuum {M}echanics.
\newblock pages 266--281, 1959.

\bibitem{Nol1972}
W.~Noll.
\newblock A new mathematical theory of simple materials.
\newblock {\em Arch. Rational Mech. Anal.}, 48(1):1--50, Jan. 1972.

\bibitem{Nol1978}
W.~Noll.
\newblock A {G}eneral {F}ramework for {P}roblems in the {S}tatics of {F}inite
  {E}lasticity.
\newblock In {\em Contemporary Developments in Continuum Mechanics and Partial
  Differential Equations, Proceedings of the International Symposium on
  Continuum Mechanics and Partial Differential Equations}, pages 363--387.
  Elsevier, 1978.

\bibitem{Olv1993}
P.~J. Olver.
\newblock {\em Applications of {L}ie {G}roups to {D}ifferential {E}quations},
  volume 107 of {\em Graduate Texts in Mathematics}.
\newblock Springer-Verlag, New York, second edition, 1993.

\bibitem{Pea1956}
C.~E. Pearson.
\newblock General theory of elastic stability.
\newblock {\em Quarterly of Applied Mathematics}, 14(2):133--144, July 1956.

\bibitem{Pen1970}
R.~W. Penn.
\newblock Volume changes accompanying the extension of rubber.
\newblock {\em Transactions of the Society of Rheology}, 14(4):509--517, Dec.
  1970.

\bibitem{PC1991}
P.~Podio-Guidugli and G.~V. Caffarelli.
\newblock Surface interaction potentials in elasticity.
\newblock In {\em Mechanics and Thermodynamics of Continua}, pages 345--385.
  Springer Berlin Heidelberg, 1991.

\bibitem{PV1989}
P.~Podio-Guidugli and M.~Vianello.
\newblock Constraint manifolds for isotropic solids.
\newblock {\em Archive for Rational Mechanics and Analysis}, 105(2):105--121,
  June 1989.

\bibitem{Rou1980}
P.~Rougée.
\newblock Formulation lagrangienne intrinsèque en mécanique des milieux
  continus.
\newblock {\em Journal de Mécanique}, 19:7--32, 1980.

\bibitem{Rou1991}
P.~Rougée.
\newblock The intrinsic {L}agrangian metric and stress variables.
\newblock {\em Finite Inelastic Deformations - Theory and Applications, IUTAM
  Symposium Hannover/Germany 199}, pages 217--226, 1991.

\bibitem{Rou1997}
P.~Rougée.
\newblock {\em Mécanique des grandes transformations}, volume~25 of {\em
  Mathématiques \& Applications (Berlin) [Mathematics \& Applications]}.
\newblock Springer-Verlag, Berlin, 1997.

\bibitem{Rou2006}
P.~Rougée.
\newblock An intrinsic {L}agrangian statement of constitutive laws in large
  strain.
\newblock {\em Computers {\&} Structures}, 84(17-18):1125--1133, June 2006.

\bibitem{Rud1991}
W.~Rudin.
\newblock {\em Functional {A}nalysis}.
\newblock International Series in Pure and Applied Mathematics. McGraw-Hill
  Inc., New York, second edition, 1991.

\bibitem{Seg1986}
R.~Segev.
\newblock Forces and the existence of stresses in invariant continuum
  mechanics.
\newblock {\em J. Math. Phys.}, 27(1):163--170, 1986.

\bibitem{SE2020}
R.~Segev and M.~Epstein.
\newblock {\em Geometric {C}ontinuum {M}echanics}, volume~42 of {\em ACM}.
\newblock Springer International Publishing, Birkhäuser Basel, 2020.

\bibitem{Sew1967}
M.~Sewell.
\newblock On configuration-dependent loading.
\newblock {\em Arch. Rational Mech. Anal.}, 23:327--351, 1967.

\bibitem{Sew1965}
M.~J. Sewell.
\newblock On the calculation of potential functions defined on curved
  boundaries.
\newblock {\em Proc. Roy. Soc. London Ser. A}, 286:402--411, 1965.

\bibitem{SM1984}
J.~C. Simo and J.~E. Marsden.
\newblock Stress tensors, {R}iemannian metrics and the alternative descriptions
  in elasticity.
\newblock In {\em Trends and applications of pure mathematics to mechanics
  ({P}alaiseau, 1983)}, volume 195 of {\em Lecture Notes in Phys.}, pages
  369--383. Springer, Berlin, 1984.

\bibitem{Sou1964}
J.-M. Souriau.
\newblock {\em Géométrie et relativité}.
\newblock Enseignement des Sciences, VI. Hermann, Paris, 1964.

\bibitem{Ste2015}
P.~Steinmann.
\newblock {\em Geometrical {F}oundations of {C}ontinuum {M}echanics}, volume~2
  of {\em Lecture Notes in Applied Mathematics and Mechanics}.
\newblock Springer, Heidelberg, 2015.
\newblock An application to first- and second- order elasticity and
  elasto-plasticity.

\bibitem{SH1997}
H.~Stumpf and U.~Hoppe.
\newblock The application of tensor algebra on manifolds to nonlinear continuum
  mechanics—invited survey article.
\newblock {\em Z. Angew. Math. Mech.}, 77(5):327--339, 1997.

\bibitem{TN1965}
C.~Truesdell and W.~Noll.
\newblock The non-linear field theories of mechanics.
\newblock In {\em Handbuch der {P}hysik, {B}and {III}/3}, pages 1--602.
  Springer-Verlag, Berlin, 1965.

\bibitem{Via2013}
M.~Vianello.
\newblock Internal constraints in finite elasticity: Manifolds or not.
\newblock {\em Journal of Elasticity}, 114(2):197--211, Feb. 2013.

\bibitem{WT1973}
C.~C. Wang and C.~Truesdell.
\newblock {\em Introduction to {R}ational {E}lasticity}.
\newblock Noordhoff International Publishing, Leyden, 1973.
\newblock Monographs and Textbooks on Mechanics of Solids and Fluids: Mechanics
  of Continua.

\bibitem{Zdu2019}
A.~Zdunek.
\newblock On purely mechanical simple kinematic internal constraints.
\newblock {\em Journal of Elasticity}, 139(1):123--152, Sept. 2019.

\end{thebibliography}
\end{document}